\input ./style/arxiv-general.cfg
\documentclass[aap,MSNbibl,nameyear,dvips]{arximspdf}
\makeatletter
   \@ifpackageloaded{graphicx}{}{\usepackage{graphicx}}
\makeatother

\doi{10.1214/14-AAP1088}
\volume{26}
\issue{1}
\pubyear{2016}
\firstpage{186}
\lastpage{215}
\docsubty{FLA}

\makeatletter
\newcommand{\rrvert}{\vert}
\newcommand{\llvert}{\vert}
\renewcommand{\mid}{|}
\newtheorem{theorem}{Theorem}[section]
\newtheorem{proposition}{Proposition}[section]
\newtheorem{lemma}{Lemma}[section]
\newtheorem{corollary}{Corollary}[section]
\newproclaim{definition}{Definition}[section]
\newproclaim{remark}{Remark}[section]
\newproclaim{example}{Example}[section]
\newcommand{\E}{\mathbb E}
\newcommand{\R}{\mathbb R}
\newcommand{\sA}{\mathcal A}
\newcommand{\sB}{\mathcal B}
\newcommand{\sP}{\mathcal P}
\newcommand{\sM}{\mathcal M}
\newcommand{\sU}{\mathcal U}
\newcommand{\uU}{\varrho}
\makeatother

\begin{document}
\begin{frontmatter}

\title{Gambling in contests with random initial law}
\runtitle{Contests with random initial law}

\begin{aug}
\author[A]{\fnms{Han}~\snm{Feng}\ead[label=e1]{H.Feng@warwick.ac.uk}}
\and
\author[A]{\fnms{David}~\snm{Hobson}\corref{}\ead
[label=e2]{D.Hobson@warwick.ac.uk}}
\runauthor{H. Feng and D. Hobson}
\affiliation{University of Warwick}
\address[A]{Department of Statistics\\
University of Warwick\\
Coventry, CV4 7AL\\
United Kingdom\\
\printead{e1}\\
\phantom{E-mail: }\printead*{e2}}
\end{aug}

%
\received{\smonth{5} \syear{2014}}
%
\revised{\smonth{11} \syear{2014}}

%
\begin{abstract}
This paper studies a variant of the contest model introduced in
Seel and Strack [\textit{J. Econom. Theory} \textbf{148} (2013) 2033--2048].
In the Seel--Strack contest, each agent or contestant
privately observes a Brownian motion, absorbed at zero, and chooses
when to stop
it. The winner of the contest is the agent who stops at the highest value.
The model assumes that all the processes start from a common value
$x_0>0$ and
the symmetric Nash equilibrium is for each agent to utilise a stopping rule
which yields a randomised value for the stopped process. In the two-player
contest, this randomised value has a uniform distribution on $[0,2x_0]$.

In this paper, we consider a variant of the problem whereby the
starting values
of the Brownian motions are independent, nonnegative random variables that
have a common law $\mu$. We consider a two-player contest and prove the
existence and uniqueness of a symmetric Nash equilibrium for the
problem. The
solution is that each agent should aim for the target law $\nu$, where
$\nu$ is
greater than or equal to $\mu$ in convex order; $\nu$ has an atom at
zero of the
same size as any atom of $\mu$ at zero, and otherwise is atom free; on
$(0,\infty)$ $\nu$ has a decreasing density; and the density of $\nu
$ only
decreases at points where the convex order constraint is binding.
\end{abstract}

%
\end{frontmatter}

\section{Introduction}

Seel and Strack (\citeyear{SeelStrack2013}) introduced a contest model in which each
agent chooses a stopping rule to stop a privately observed stochastic
process. The contestant who stops her process at the highest value wins
a prize. The objective of each agent is not
to maximise the expected stopping value, but rather to maximise the
probability that her stopping value is the highest amongst the set
of stopping values of all the contestants.

The Seel--Strack contest is a stylised model of a contest between
agents in which agents
compete to win a single prize. Examples include certain internet casino
games (where competitors ante a fixed amount, then play independently
with notional funds, with the prize awarded to the contestant with the
highest notional fortune at the end of the game), competitions between
fund managers
(where each manager aims to outperform all the others in order to
obtain more funds to invest
over the next time period) and competitions between company CEOs (only
the most successful of whom, in relative terms, will be offered an
executive position at a larger company); see
\citet{SeelStrack2013} for further details and examples. The key
feature of these contests is that the cost in terms of effort is not
dependent on the riskiness of the strategy, and that agents do not earn
rewards from the value of their entry into the contest, but only from
the relative value or rank order of the entry.


The modelling assumptions of \citet{SeelStrack2013} include the
fact that contestants are unable to observe both the realisations and
the stopping times of their rivals and the fact that the privately
observed stochastic processes are independent realisations of a
drifting Brownian motion absorbed at zero. As argued by \citet
{FengHobson14}, the setting can be generalised to allow for
independent realisations of any nonnegative, time-homogeneous
diffusion process since a change of scale and time reduces the problem
to the martingale case of Brownian motion without drift. Henceforth, we
will concentrate on the case of drift-free Brownian motion, absorbed at
zero. We will focus on the two-player case.

\citet{SeelStrack2013} studied the symmetric case where the
contestants observe processes which all
start from the same strictly positive constant.
In this paper, we will discuss an extension of the Seel--Strack problem
to randomised initial values, whereby the starting values of the
Brownian motions
are drawn independently from the (commonly known) distribution $\mu$,
where $\mu$ is any integrable probability measure on $\R^+$. This might
correspond to casino gamblers who are obliged to participate in a
minimum number of bets before closing out their position, fund managers
with different initial portfolios (which are unknown to competing
agents) or to newly installed CEOs taking positions in companies of
different strengths.

A very attractive feature of the
Seel--Strack contest model is that it has an explicit symmetric Nash
equilibrium, as constructed in \citet{SeelStrack2013}.
Stopping rules are identified with target laws for the entry into the
contest and in equilibrium players use randomised strategies, so that
the level
at which the player should stop is stochastic. Moreover, the set of
values at which the agent should stop forms an interval which is
bounded above. In the two-player case, the target law is a uniform
distribution (so that if the initial wealth of both agents is $x$ then
the target law has constant density $1/2x$ on $[0,2x]$); see
Example~\ref{eguni} below. Our results extend \citet
{SeelStrack2013} to the case where the initial values of the Brownian
motions are independent draws from a common initial law. As in the
Seel--Strack contest, stopping rules are identified with target laws,
and in equilibrium players use randomised strategies, so that the level
at which the player should stop is stochastic. Now, however, the set of
values at which the agent should stop forms an interval which can be
unbounded above (if the initial law has unbounded support). Since
terminal laws are obtained from stopping a nonnegative martingale, it
is clear that any attainable terminal law has a mean which is equal to
or less than the mean of the initial law, and it is natural to expect
that an optimal terminal law has highest mean possible, or equivalently
we expect that for an optimal stopping rule the mean of the terminal
law should equal that of the initial law. Then the candidate terminal
laws are precisely those which are greater than or equal to the initial
law in convex order. In fact, we show that the optimal law has the twin
properties that (with the possible exception of an atom at zero of the
same size as the inital law) the target law has a density which is
decreasing, and the density decreases only at those levels where the
convex order constraint, expressed via the potential of the
distributions, is binding. We have two main results: first, we show
that any distribution with these properties characterises a symmetric
Nash equilibrium for the problem (Theorem~\ref{thmcharacterisation});
and second we show that for any initial law there is exactly one target
law with these properties (Theorem~\ref{thmSingleton}), and hence
there is a unique symmetric Nash equilibrium for the problem.

The Seel--Strack contest is a recent addition to the literature in
economics and has not yet been greatly studied.
However, as \citet{SeelStrack2013} emphasise, there are strong
connections between their contest model and all-pay auctions [see,
e.g., \citet{BayeKovenockVries1996}], wars of attrition [\citet
{HendricksWeissWilson1988}] and silent timing games [\citet
{ParkSmith2008}]. Indeed in the setting of $n$ exchangeable agents
with fixed initial wealth (equal to $1/n$), the solution of the
Seel--Strack contest is identical to that of $n$-agents in an all-pay
auction with cost equal to bid size. In both situations, agents
randomise their strategy in order to introduce uncertainty into the
value they enter into the contest/auction. This makes it more difficult
(and in equilibrium, impossible) for opponents to take advantage of
knowledge of the distribution of the entry. If in the two-player
Seel--Strack contest an agent chooses an entry which is a point mass
(at the initial wealth $x$), then the opponent can play until his
wealth is $x+\varepsilon$ or he goes bankrupt; thus winning the contest
with probability $x/(x + \varepsilon)$. Since $\varepsilon$ can be chosen
arbitrarily the opponent can choose a strategy for which the
probability of him winning is arbitrarily close to unity, and thus the
strategy of the first player cannot be optimal.

In contrast to the Seel--Strack contest, the all-pay auction is much
studied. The all-pay auction is used as a model of technological
competition, political lobbying and job promotion (with effort).
Several generalisations of the all-pay auction have been studied.
Sometimes in these auctions one or more of the agents has a headstart
[\citet{Konrad2002}], perhaps from prior effort, or from being the
incumbent in an election campaign.
In recent work, \citet{Seel2013} studies the all-pay auction with
random head starts, which is the direct analogue of the problem studied
here. Seel finds the Nash equilibrium in an asymmetric, two player
all-pay auction where one of the players has a random headstart, and
the ideas can be extended to the symmetric case where both players have
random headstarts.

Relative to the main results of \citet{Seel2013} on all-pay
auctions with random headstarts, we find that the symmetric Nash
equilibrium in the
Seel--Strack contest with random initial law is more subtle. In the
symmetric two-player all-pay auction with random headstart, the target
law for the total bid has a density taking values in $\{0,1\}$. In
particular, there may be gaps in the support of the distribution. In
contrast, in the symmetric two-player Seel--Strack contest with random
initial law, we find that the support of the target law has no holes,
and that the density is monotonic decreasing. This indicates that the
Seel--Strack contest and the all-pay auction are perhaps less closely
linked than might be expected from considering the base case alone. In
the case without headstarts/with point mass initial law, the Nash
equlibria in the two settings are identical, but this does not carry
over to generalisations of the two problems. Even in the standard
setting we learn that the fundamental property of the uniform
distribution which leads to optimality in~the two-player Seel--Strack
contest is that it has a decreasing density, whereas in the context of
the all-pay auction it is the fact that the uniform distribution has a
density which takes values in $\{0,1\}$.

\citet{SeelStrack2013} solved their problem by proposing a candidate
value function for the problem, and then verifying that this candidate
is a martingale under an optimal stopping rule for each agent. We
solve the problem in a different way using a
Lagrangian approach.
Our solution methods allow us to characterise the target laws which
correspond to a candidate equilibrium quite easily, but some effort is
required to show that there is a measure with these characteristics,
and that for a given initial law, this measure is unique.

The paper is structured as follows. In Section~\ref{secModel}, we
introduce the mathematical model of the contest.
We will see that a Nash equilibrium is identified with a pair of
probability measures. In Section~\ref{secCondiSingl}, we state the
main theorems which characterise
the unique symmetric Nash equilibrium and give some examples.
A proof of the characterisation theorem is given in Section~\ref
{secsufficiency}.
Following some preliminary results in Section~\ref{secPreliminary},
which are potentially of independent interest, Section~\ref{secExistence}
includes an explicit construction of the symmetric Nash equilibrium
in the case where the starting random variable takes only a finite
number of distinct
values, and then an extension of the existence result to general measures.
Uniqueness of the equilibrium is proved in Section~\ref{secUnique}.


\section{The model}
\label{secModel}

Consider a contest between two agents. Agent $i \in\{1,2\}$
privately observes the continuous-time realisation of a Brownian motion
$X^{i}=(X_{t}^{i})_{t\in\mathbb{R}^{+}}$ absorbed at zero, where the
processes $X^{i}$ are independent. Seel and Strack assume $X^i_0 = x_0$
for some positive real number $x_0$ which does not depend on~$i$. The
innovation in this paper is that we
assume $X_{0}^{i} \sim\mu$, where $\mu$ is the law of a nonnegative
random variable with finite mean $\overline{\mu} \in(0,\infty)$. The
values of $(X^i_0)_{i=1,2}$ are assumed to be independent draws from
the law $\mu$.

Let $\mathcal{F}_{t}^{i}=\sigma(\{X_{s}^{i}\dvtx s \leq t\})$ and set
$\mathbb{F}^{i}=(\mathcal{F}_{t}^{i})_{t\geq0}$.
A strategy of agent $i$ is a \mbox{$\mathbb{F}^{i}$-}stopping time $\tau^{i}$.
Since zero is absorbing for $X^{i}$, without loss of generality we
may restrict attention to $\tau^{i}\leq H_{0}^{i}=\inf\{t\geq
0\dvtx X_{t}^{i}=0\}$.
Both the process $X^{i}$ and the stopping time $\tau^{i}$ are private
information to agent $i$. That is, $X^{i}$ and $\tau^{i}$ cannot
be observed by the other agent.

The agent who stops at the highest value wins a prize, which we normalise
to one without loss of generality. In the case of a tie in which both
agents stop at the equal highest value, we assume that each of them
wins $\theta$, where $\theta\in[0,1)$. Therefore, player $i$ with
stopping value $X_{\tau^{i}}^{i}$ receives payoff
\[
\mathbf{1}_{\{X_{\tau^{i}}^{i} > X_{\tau^{3-i}}^{3-i}\}}+\theta\mathbf
{1}_{\{X_{\tau^{i}}^{i}=X_{\tau^{3-i}}^{3-i}\}}.
\]

\citet{SeelStrack2013} observed that since the payoffs to the agents
only depend upon $\tau^{i}$ via the distribution of $X_{\tau^{i}}^{i}$,
the problem of choosing the
optimal stopping time $\tau^{i}$ can be reduced to a problem of finding
the optimal distribution of $X_{\tau^{i}}^{i}$ or equivalently an
optimal target law.
Once we have found the optimal target law $\nu^{i}$, the remaining
work is to verify that there exists $\tau^{i}$ such\vspace*{2pt} that $X_{\tau
^{i}}^{i}\sim\nu^{i}$.
This is the classical Skorokhod
embedding problem [\citet{Skorokhod65}] for which solutions are
well known
[see \citet{Obloj2004} and \citet{Hobson2011} for surveys].
Note that there are
typically multiple solutions to the embedding problem for a given
target law, and
any of these solutions can be used to construct an optimal stopping rule.
We will identify a solution with the distribution of $X_{\tau
^{i}}^{i}$, rather than with
the stopping rule itself, so that when we talk about a unique
equilibrium the
uniqueness will refer to the target distribution and not to the
stopping rule.

We now introduce some notation which will be used throughout the paper.
Let $\sM$ be the set of integrable measures on $\R^+=[0,\infty)$ with
finite total mass, and let $\sP$ be the subset of $\sM$ consisting of
integrable probability measures on $\R^+$. Let $\delta_x \in\sP$ be
the unit mass at $x$ and let $\uU_x \in\sP$ be the uniform
distribution on $[0,2x]$ (with mean $x$).

For
$\varpi\in\sM$, let $\overline{\varpi} = \int_0^\infty x \varpi(dx)$,
and define the right-continuous distribution function $F_\varpi\dvtx
[0,\infty) \mapsto[0,\varpi(\R^+)]$ by
$ F_{\varpi}(x)=\varpi([0,x])$.
In the sequel, we occasionally want to consider $F_{\varpi}$ as a
function on $(-\infty,\infty)$ in which case we set $F_\varpi(x)=0$
for $x<0$.
Define also the call and put price functions $C_\varpi\dvtx[0,\infty)
\mapsto[0,\overline{\varpi}]$ and $P_{\varpi}\dvtx[0,\infty) \mapsto
[0,\infty)$ by
\begin{eqnarray*}
C_{\varpi}(x) & = & \int_{x}^{\infty}(y-x)
\varpi(dy) = \int_x^\infty\bigl( \varpi\bigl(\R^+
\bigr) - F_\varpi(y) \bigr) \,dy,
\\
P_{\varpi}(x) & = & \int_{0}^{x}(x-y)
\varpi(dy) = \int_0^x F_\varpi(y) \,dy.
\end{eqnarray*}
Note that $C_{\varpi}(x)-P_{\varpi}(x)=\int_{0}^{\infty}y\varpi
(dy)-x\int_{0}^{\infty}\varpi(dy) = \overline{\varpi}-x \varpi(\R^+)$
and if $\chi\in\sM$ and if $\varpi(\R^+) = \chi(\R^+)$ then
%
\begin{eqnarray}\label{eqnPcond}
&& \lim_{x\uparrow\infty} \bigl(P_{\varpi}(x)-P_{\chi}(x)
\bigr) \nonumber
\\
&&\qquad  =\lim_{x\uparrow
\infty} \bigl(P_{\varpi}(x)-x\varpi\bigl(
\R^{+} \bigr) \bigr)-\lim_{x\uparrow\infty
} \bigl(P_{\chi
}(x)-x
\chi\bigl(\R^{+} \bigr) \bigr)
\\
&&\qquad  =\overline{\chi}-\overline{\varpi}.\nonumber
\end{eqnarray}
Then $\chi$ is less than or equal to $\varpi$ in convex order (written
$\chi\preceq_{\mathrm{cx}} \varpi$) if and only if $\chi(\R^+) = \varpi(\R^+)$,
$\overline{\chi} = \overline{\varpi}$ and $C_{\chi}(x)\leq
C_{\varpi
}(x)$ for all $x\geq0$. This last condition can be rewritten in terms
of puts as $P_{\chi}(x)\leq P_{\varpi}(x)$ for all $x\geq0$.

Note that if $\chi\preceq_{\mathrm{cx}} \varpi$, then $C_\chi(0)=C_{\varpi
}(0)$ and $\chi(\{0\}) = F_{\chi}(0) = \chi(\R^+) + C_\chi'(0+)
\leq
\varpi(\R^+) + C'_\varpi(0+) = F_{\varpi}(0) = \varpi( \{0\})$.

Suppose $\chi, \varpi\in\sP$, then it is well known [see, e.g.,
\citet{ChaconWalsh1976}] that for Brownian motion $X$ with $X_0
\sim\chi$, there exists a stopping time $\tau$ for which $(X_{t
\wedge
\tau})_{t \geq0}$ is uniformly integrable and $X_\tau\sim\varpi$ if
and only if $\chi\preceq_{\mathrm{cx}} \varpi$.
In our context, we do not necessarily want to insist on uniform
integrability, but rather that the stopping time occurs before the
first hit of $X$ on zero.

\begin{lemma} Suppose $\chi, \varpi\in\sP$.
Let $X$ be a Brownian motion absorbed at zero with initial law $\chi$.
Then there exists a stopping time $\tau$ with $X_\tau\sim\varpi$ if
and only if $P_{\chi}(x) \leq P_{\varpi}(x)$ for all $x \geq0$.
\end{lemma}

\begin{pf}
Since $(X_{t \wedge H_0})_{t \geq0}$ is a nonnegative
supermartingale, by a conditional version of Jensen's inequality,
\[
\E\bigl[ (x - X_{\tau})^+ \mid{\mathcal F}_0 \bigr] \geq
\bigl(x - \E[X_\tau\mid{\mathcal F}_0] \bigr)^+ \geq(x -
X_0)^+
\]
and then
\[
P_{\tau}(x) = \E\bigl[(x - X_{\tau})^+ \bigr] = \E\bigl[\E
\bigl[ (x - X_{\tau})^+ \mid{\mathcal F}_0 \bigr] \bigr] \geq
\E\bigl[(x - X_0)^+ \bigr] = P_{\chi}(x).
\]
Conversely, the existence follows from results of \citet{Rost71}.
\end{pf}

\begin{definition} Given $\mu\in\sM$,
we say that $\nu\in\sM$ is \emph{weakly admissible} (with respect to
$\mu$) if $\nu(\R^+) = \mu(\R^+)$ and $P_{\nu}(x) \geq P_{\mu
}(x)$ for
all $x \geq0$.

We say that $\nu$ is \emph{strongly admissible} (with respect to $\mu$)
if $\nu$ is weakly admissible and $\overline{\nu}= \overline{\mu}$.
\end{definition}

Note that if $\nu$ is weakly admissible then necessarily $\nu(\{0 \})
\geq\mu(\{0 \})$ and $\overline{\nu} \leq\overline{\mu}$ by
(\ref
{eqnPcond}). If $\nu$ is strongly admissible, then $\mu\preceq_{\mathrm{cx}}
\nu$.

These definitions are motivated by the fact that if $X$ is Brownian
motion with $X_0 \sim\mu\in\sP$ and if $\nu\in\sP$ is weakly
admissible with respect to $\mu$ then there exists $\tau\leq H_0:=
\inf\{ u>0\dvtx X_u = 0 \}$ such that $X_\tau\sim\nu$, and if $\nu$ is
strongly admissible then there exists $\tau$ such that $X_\tau\sim
\nu
$ and $(X_{t \wedge\tau})_{t \geq0}$ is uniformly integrable.

\begin{definition}
The pair of weakly admissible measures $(\nu^{1},\nu^{2})$ with $\nu^i
\in\sP$ forms a \emph{Nash equilibrium}
if, for each $i\in\{1,2\}$, given that the other agent $j=3-i$
uses a stopping rule $\tau^{j}$ such that $X_{\tau^{j}}^{j}\sim\nu^{j}$,
then any stopping rule $\tau^{i}$ such that $X_{\tau^{i}}^{i}\sim\nu^{i}$
is optimal.

If $(\nu,\nu)$ forms a Nash equilibrium, then the equilibrium is symmetric.
\end{definition}

A Nash equilibrium may be characterised as follows. Let $V^i_{\chi
,\varpi}$ denote the value of the game to player $i$ if Player 1 uses a
stopping rule which yields law $\chi$ for $X^1_{\tau^1}$ and Player 2
uses a stopping rule which yields law $\varpi$ for $X^2_{\tau^2}$.
It follows that
\begin{eqnarray*}
V^1_{\chi, \varpi} & = & \int_{[0,\infty)} \chi(dx)
\bigl\{ F_{\varpi}(x-) + \theta\bigl(F_{\varpi}(x) -
F_\varpi(x-) \bigr) \bigr\},
\\
V^2_{\chi, \varpi} & = & \int_{[0,\infty)} \varpi(dx)
\bigl\{ F_{\chi}(x-) + \theta\bigl(F_{\chi}(x) -
F_\chi(x-) \bigr) \bigr\}.
\end{eqnarray*}
Then $(\sigma,\nu)$ is a Nash equilibrium if
$V^1_{\sigma, \nu} = \sup_{\pi} V^1_{\pi,\nu}$ and $V^2_{\sigma,\nu} =
\sup_{\pi} V^2_{\sigma,\pi}$, where the suprema are taken over the set
of weakly admissible measures. A symmetric Nash equilibrium is a weakly
admissible measure $\nu$ such that $V^1_{\nu,\nu} = \sup_{\pi}
V^1_{\pi,\nu}$ and $V^2_{\nu, \nu} = \sup_{\pi} V^2_{\nu, \pi}$, where again
the supremum is taken over weakly admissible measures $\pi$. Note that
since $V^1_{\chi, \varpi} = V^2_{\varpi,\chi}$ the definition can be
simplified to a symmetric Nash equilibrium is a weakly admissible
measure $\nu$ such that $V^1_{\nu,\nu} = \sup_{\pi} V^1_{\pi,\nu
}$ from
which is follows that $\sup_{\pi} V^2_{\nu, \pi} = \sup_{\pi}
V^1_{\pi,\nu} = V^1_{\nu,\nu} = V^2_{\nu,\nu}$ and $\nu$ is also
optimal for
Player 2.

Our paper investigates the existence and
uniqueness of a symmetric Nash equilibrium. It seems natural that
a Nash equilibrium is symmetric, since the contest is symmetric in
the sense that each agent observes a martingale process started from
the same law $\mu$. Then simple arguments over rearranging mass can
be used to show that it is never optimal for two agents to put mass
at the same positive point~$x$, and further that any symmetric Nash
equilibrium must be strongly admissible. The proof of Theorem \ref
{thmAtomfree}
is in the \hyperref[secProveAtomfree]{Appendix}.

\begin{theorem} \label{thmAtomfree}
Suppose $(\nu,\nu)$ is a Nash equilibrium. Then $\nu$ is strongly
admissible, $F_{\nu}(x)$ is
continuous on $(0,\infty)$ and $F_{\nu}(0)=F_{\mu}(0)$.
\end{theorem}

\section{Main results and examples}\label{secCondiSingl}

In this section, we describe the main results concerning existence and
uniqueness of a symmetric Nash equilibrium
for the contest, and give examples.

\begin{definition}\label{defCondi}
Suppose $\mu\in\sM$.
Let $\mathcal{A}_{\mu}^{*} \subseteq\sM$ be the set of measures
$\nu$
satisfying:
\begin{longlist}[(iii)]
\item[(i)] $\nu(\R^+) = \mu(\R^+)$, $F_{\nu}(0)=F_{\mu}(0)$,
$F_{\nu}$
is continuous
on $(0,\infty)$, $\overline{\nu}=\overline{\mu}$,
and $C_{\nu}(x)\geq C_{\mu}(x)$ for all $x\geq0$;
\item[(ii)] $F_{\nu}(x)$ is concave on $[0,\infty)$;
\item[(iii)] if $C_{\nu}(x)>C_{\mu}(x)$ on some interval $\mathcal
{J}\subset[0,\infty)$
then $F_{\nu}(x)$ is linear on $\mathcal{J}$.
\end{longlist}
\end{definition}

The two main results in this article are a theorem which characterises
symmetric Nash equilibria, and a theorem which proves that a symmetric
Nash equilibrium exists and is unique.

\begin{theorem}
\label{thmcharacterisation} If $\mu\in\sP$, then $(\nu^{*},\nu
^{*})$ is
a symmetric Nash equilibrium for the problem if and only if $\nu
^{*}\in
\mathcal{A}_{\mu}^{*}$.
\end{theorem}

\begin{theorem}
\label{thmSingleton} For $\mu\in\sM$,
$\llvert\mathcal{A}_{\mu}^{*}\rrvert=1$. In particular, if $\mu\in
\sP$ then there
exists a unique symmetric Nash equilibrium for the problem.
\end{theorem}

Before proving Theorems \ref{thmcharacterisation} and \ref
{thmSingleton}, we present some examples.

\begin{example}\label{eguni}
Recall that $\uU_{x}=\mathcal{U}[0,2x]$, where $\mathcal{U}$ stands
for the continuous uniform distribution. Suppose $\mu\in\sP$ satisfies
$C_{\mu}\leq C_{\uU_{\overline{\mu}}}$. Then it is easy to see that
$\uU
_{\overline{\mu}}\in\mathcal{A}_{\mu}^{*}$,
and thus $(\uU_{\overline{\mu}},\uU_{\overline{\mu}})$ is the unique
symmetric Nash equilibrium
for the problem. In the case $\mu= \delta_0$, this is the \citet
{SeelStrack2013} result.
\end{example}

\begin{example}
Suppose $\mu\in\sP$ is atom-free, except perhaps for an atom at
zero. Set
$b_{\mu}=\sup\{ x\dvtx F_{\mu}(x)<1 \} $.
If $F_{\mu}$ is concave on $[0,b_{\mu}]$, then $\mu\in\mathcal
{A}_{\mu}^{*}$
and $(\mu,\mu)$ is the unique symmetric Nash equilibrium for the
problem. If $F_{\mu}$ is convex on $[0,b_{\mu}]$ and $F_{\mu}(0)=0$,
then it can be verified that $C_{\mu}\leq C_{\uU_{\overline{\mu}}}$
(see Proposition
\ref{propBound} for a detailed proof),
and thus $(\uU_{\overline{\mu}},\uU_{\overline{\mu}})$
is the unique symmetric Nash equilibrium
for the problem.
\end{example}

\begin{example}[(Beta distribution)] Suppose $\mu$ is a Beta distribution on $[0,1]$
with shape
parameters $\alpha=2$ and $\beta=3$, that is $\mu=\operatorname{Beta}(2,3)$.
Then, the mean of $\mu$ is $2/5$, $C_{\mu}(x)= (\frac
{3}{5}x^{5}-2x^{4}+2x^{3}-x+\frac{2}{5} )\mathbf{1}_{\{x\leq1\}}$,
and $F_{\mu}(x)=\min\{3x^{4}-8x^{3}+6x^{2},1\}$. Then $F_{\mu}(x)$
is convex on $(0,\frac{1}{3})$ and concave on $(\frac{1}{3},1)$,
or equivalently the density $f_\mu= F_{\mu}'$ is increasing on
$[0,1/3]$ and decreasing on $[1/3,1]$. Hence, $\mu$ itself is not a
candidate Nash equilibrium. Instead we expect to find a symmetric Nash
equilibrium $\nu$ with a constant density $f_\nu(x) = 2c_1$ on
$[0,c_2]$, and $f_\nu(x) = f_\mu(x)$ on $[c_2,1]$, where $c_1$ and
$c_2$ are constants to be determined. Since $\nu$ has constant density
$2c_1$ on $[0,c_2]$
it follows that $C_{\nu}(x)=c_{1}x^{2}-x+\frac{2}{5}$
on this interval. Moreover,
$c_{1}$ and $c_{2}$ satisfy $C_{\mu}(c_{2})=C_{\nu
}(c_2)=c_{1}c_{2}^{2}-c_{2}+\frac{2}{5}$
and $C_{\mu}'(c_{2})=C_\nu'(c_2)=2c_{1}c_{2}-1$. Solving the system of
equations,
we obtain $c_{1}=\frac{4\sqrt{10}+140}{243}$ and $c_{2}=\frac
{10-\sqrt{10}}{9}$.
For this choice of $c_1,c_2$ it follows that
$\nu\in\mathcal{A}_{\mu}^{*}$. Thus, $(\nu,\nu)$
is the unique symmetric Nash equilibrium for the problem.
\end{example}

\begin{example}[(Atomic measure)]
Suppose that $\mu=\frac{1}{2}\delta_{1-\varepsilon
}+\frac
{1}{2}\delta_{1+\varepsilon}$,
where $\varepsilon\in(0,1)$. Then
\[
C_{\mu}(x)=(1-x)\cdot\mathbf{1}_{x\in[0,1-\varepsilon)}+\tfrac
{1}{2}(1+
\varepsilon-x)\cdot\mathbf{1}_{x\in[1-\varepsilon,1+\varepsilon)}.
\]

Suppose $\varepsilon\in(0,1/2]$. Then $C_{\mu}\leq C_{\uU_1}$,
where $\uU_1=\mathcal{U}[0,2]$, and
then $(\uU_1,\uU_1)$ is the unique symmetric Nash equilibrium
for the problem.

Now suppose $\varepsilon\in(1/2,1)$. Define the function
$C$ by
\[
C(x)=\frac{x^{2}-8(1-\varepsilon)(x-1)}{8(1-\varepsilon)}\cdot\mathbf
{1}_{x\in
[0,2(1-\varepsilon))}+\frac{x^{2}-8\varepsilon x+16\varepsilon
^{2}}{8(3\varepsilon
-1)}\cdot
\mathbf{1}_{x\in[2(1-\varepsilon),4\varepsilon)},
\]
and let $\nu$ be given by $C_\nu(x)=C(x)$.
Then $\nu\in\mathcal{A}_{\mu}^{*}$.
Hence, $(\nu,\nu)$ is the unique symmetric Nash equilibrium for the
problem if $\varepsilon\in(1/2,1)$.
\end{example}

\section{Sufficiency}\label{secsufficiency}

\mbox{}
\begin{pf*}{Proof of reverse implication of Theorem \ref{thmcharacterisation}}
We show that if $\mu\in\sP$ and $\nu\in\sA^*_\mu$ then $(\nu,
\nu)$
is a symmetric Nash equilibrium. The statement that if $(\nu,\nu)$ is a
symmetric Nash equilibrium then $\nu\in\sA^*_\mu$ is given in the \hyperref[secProveAtomfree]{Appendix}.\vspace*{1pt}

Given $\mu\in\sP$, define the classes of measures
$\sA^w_\mu= \{ \nu\in\sM$, $\nu$ weakly admissible with respect to
$\mu\}$ and
$\sA^s_\mu= \{ \nu\in\sM$, $\nu$ strongly admissible with respect
to~$\mu\}$.
Recall also the definition of $\sA^*_\mu$ in (\ref{defCondi}).

Using the properties of Theorem~\ref{thmAtomfree} we find that a
symmetric Nash equilibrium is identified with a measure $\nu^{*}\in
\sA
^s_{\mu}$
with the property that, for any $\nu\in\sA^w_{\mu}$, $V^1_{\nu
^*,\nu^*}
\geq V^1_{\nu,\nu^*}$. Since again by Theorem~\ref{thmAtomfree}
we have that $\nu^*$ has no atoms on $(0,\infty)$, for a symmetric Nash
equilibrium we must have that for any $\nu\in\sA^w_\mu$
%
\begin{eqnarray}\label{eqDefnNash}
&& \int_{(0,\infty)}F_{\nu^{*}}(x)\nu^{*}(dx)+\theta
F_{\nu
^{*}}(0)F_{\nu
^{*}}(0)
\nonumber\\[-8pt]\\[-8pt]\nonumber
&&\qquad \geq\int_{(0,\infty)}F_{\nu^{*}}(x)\nu(dx)+  \theta
F_{\nu
^{*}}(0)F_{\nu}(0).
\end{eqnarray}

Fix $\nu^*\in\sA_{\mu}^{*}$ and suppose $\nu\in\sM$.
Suppose that $\lambda$, $\gamma$ and $\zeta$ are finite
constants and $\eta$ is a measure on $(0,\infty)$, and suppose
$\lambda
$ and $\zeta$ are nonnegative.
Define
%
\begin{eqnarray}
&& \mathcal{L}_{\nu^{*}}(\nu;\lambda,\gamma,\zeta,\eta)
\nonumber
\\
& & \qquad=\int_{(0,\infty)}F_{\nu^{*}}(x)\nu(dx)+\theta
F_{\nu
^{*}}(0)F_{\nu}(0)+\lambda\biggl(\overline{\mu}-\int
_{(0,\infty
)}x\nu(dx) \biggr)
\nonumber
\\
&&\quad\qquad{}+\gamma\biggl(1-\int_{(0,\infty)}
\nu(dx)-F_{\nu
}(0) \biggr)+\zeta\bigl(F_{\nu}(0)-F_{\mu}(0)
\bigr)
\nonumber
\\
&&\quad\qquad{}+\int_{(0,\infty)} \bigl(P_{\nu}(z)-P_{\mu
}(z)
\bigr)\eta(dz)\label{eqOptimLag1}
\\
&&  \qquad=\int_{(0,\infty)} \biggl(F_{\nu^{*}}(x)-\lambda x-
\gamma+\int_{(x,\infty)}(z-x)\eta(dz) \biggr)\nu(dx)
\nonumber
\\
&&\quad\qquad{}+ \biggl(\theta F_{\nu^{*}}(0)-\gamma+\zeta+\int
_{(0,\infty)}z\eta(dz) \biggr)F_{\nu}(0)
\nonumber
\\
&&\quad\qquad{}+\lambda\overline{\mu}+\gamma-\int_{(0,\infty
)}P_{\mu
}(z)
\eta(dz)-\zeta F_{\mu}(0),\label{eqOptimLag2}
\end{eqnarray}
where we use
\begin{eqnarray*}
\int_{(0,\infty)}P_{\nu}(z)\eta(dz) & = & \int
_{(0,\infty)} \eta(dz) \biggl[ z \nu\bigl( \{0\} \bigr) + \int
_{(0,z)} (z-x) \nu(dx) \biggr]
\\
& = & F_{\nu}(0)\int_{(0,\infty)} z \eta(dz) + \int
_{(0,\infty)} \nu(dx) \int_{(x,\infty)} (z-x) \eta(dz).
\end{eqnarray*}
Equivalently, we have
\begin{eqnarray}
&& \int_{(0,\infty)}  F_{\nu^{*}}(x)\nu(dx)+\theta
F_{\nu
^{*}}(0)F_{\nu
}(0)
\nonumber
\\
&&\qquad  =\mathcal{L}_{\nu^{*}}(\nu;\lambda,\gamma,\zeta,\eta)-\lambda(
\overline{\mu}-\overline{\nu})- \gamma\bigl(1 - \nu\bigl(\R^+ \bigr)
\bigr)-
\zeta\bigl(F_{\nu
}(0)-F_{\mu}(0) \bigr)
\nonumber
\\
&&\quad\qquad{}-\int_{(0,\infty)} \bigl(P_{\nu}(z)-P_{\mu}(z)
\bigr)\eta(dz).
\nonumber
\end{eqnarray}
Now suppose $\nu\in\sA^w_\mu$. Then since $\lambda\geq0$, $\zeta
\geq
0$ and $\eta(dz)\geq0$ and since $\nu\in\sA^w_\mu$ implies that
$\nu
(\R^+)= \mu(\R^+)=1$, $F_{\nu}(0)\geq F_{\mu}(0)$
and $P_{\nu}(z)\geq P_{\mu}(z)$ $\forall z\geq0$, we find
%
\begin{equation}
\int_{(0,\infty)} F_{\nu^{*}}(x)\nu(dx)+\theta
F_{\nu
^{*}}(0)F_{\nu}(0) \leq\mathcal{L}_{\nu^{*}}(\nu;
\lambda,\gamma,\zeta,\eta).\label
{eqpayLeqLag}
\end{equation}
Furthermore, if $\eta$ is such that $\eta(\mathcal{J})=0$ for every
interval $\mathcal{J}$ for which $P_{\nu^*}(z) > P_{\mu}(z)$ on
$\mathcal{J}$,
then since $\nu^*$ has unit mass and mean $\overline{\mu}$ and since
$F_{\nu^{*}}(0)=F_{\mu}(0)$,
%
\begin{equation}
\int_{(0,\infty)}F_{\nu^{*}}(x)\nu^{*}(dx)+\theta
F_{\nu
^{*}}(0)F_{\nu
^{*}}(0)=\mathcal{L}_{\nu^{*}} \bigl(
\nu^{*};\lambda,\gamma,\zeta,\eta\bigr).\label
{eqpayEqLag}
\end{equation}

Define $b^{*}=\sup\{ x\dvtx F_{\nu^{*}}(x)<1 \} $ so that
$b^{*}\leq\infty$.
Since $F_{\nu^{*}}$ is continuous and concave, it is absolutely
continuous, which means that there exists a function $f_{\nu^{*}}$
such that $F_{\nu^{*}}(x)=\int_{0}^{x}f_{\nu^{*}}(y)\,dy+F_{\nu^{*}}(0)$.
By concavity of $F_{\nu^*}$, $f_{\nu^*}$ is monotonic and we may take
it to be right-continuous. Then $f_{\nu^{*}}(x)=\int_{(x,b^{*}]}\psi
(dz)$, where the measure $\psi$
is given by $\psi((z_{1},z_{2}])=f_{\nu^{*}}(z_{1})-f_{\nu^{*}}(z_{2})$
for any $z_{1}<z_{2}$.

Let $\lambda^{*}=0$,
$\gamma^{*}=1$, $\zeta^{*}=(1-\theta)F_{\nu^{*}}(0)$
and $\eta^{*}=\psi$. Then $\int_{(0,\infty)} z \eta^*(dz) = 1 -
F_{\nu^*}(0)$.
Since $\nu^* \in\sA^*_\mu$ we have that\vspace*{1pt} if $\eta^*$ places mass on
every neighbourhood of $x$ then $P_{\nu^*}(x) = P_\mu(x)$.

Define $\Gamma$ on $[0,\infty)$ by $\Gamma(x)=\lambda^{*}x+\gamma
^{*}-\int_{(x,\infty)}(z-x)\eta^{*}(dz)$. Then, for any $x>0$,
\[
\Gamma(x)= 1 - \int_{(x,\infty)} \psi(dz) \int_x^z dy = 1 -\int_x^\infty dy\, f_{\nu^*}(y) =
F_{\nu^*}(x).
\]
Observe
that $\theta F_{\nu^{*}}(0)-\gamma^{*}+\zeta^{*} + \int_{(0,\infty
)} z
\eta^*(dz)=0$. Thus, by (\ref{eqpayLeqLag})
and (\ref{eqOptimLag2}), for $\nu\in\sA^w_\nu$,
%
\begin{eqnarray}\label{eq6a}
&& \int_{(0,\infty)}F_{\nu^{*}}(x)\nu(dx)+\theta
F_{\nu^{*}}(0)F_{\nu}(0)
\nonumber\\[-8pt]\\[-8pt]\nonumber
&&\qquad \leq\lambda^{*}\overline{\mu}+\gamma^{*}  -\int
_{(0,\infty
)}P_{\mu }(z)\eta^{*}(dz)-
\zeta^{*}F_{\mu}(0),
\end{eqnarray}
and by (\ref{eqpayEqLag}) and (\ref{eqOptimLag2}),
%
\begin{eqnarray}\label{eq6b}
&& \int_{(0,\infty)}F_{\nu^{*}}(x)\nu^{*}(dx)+\theta
F_{\nu
^{*}}(0)F_{\nu
^{*}}(0)
\nonumber\\[-8pt]\\[-8pt]\nonumber
&&\qquad =\lambda^{*}\overline{\mu}+\gamma^{*}  -\int
_{(0,\infty)}P_{\mu
}(z)\eta^{*}(dz)-
\zeta^{*}F_{\mu}(0).
\end{eqnarray}
Note $\int_{(0,\infty)}P_{\mu}(z)\eta^{*}(dz)=\int_{(0,b^{*}]}P_{\nu
^*}(z)\psi(dz) = \int_{(0,b^*]} f_{\nu^*}(z) F_{\nu^*}(z) \,dz = (1 -
F_{\nu^*}(0)^2)/2 < 1$
so that the right-hand side of (\ref{eq6a}) and (\ref{eq6b}) is
well defined and positive. Furthermore, for any $\nu\in\sA^w_{\mu}$,
\[
\int_{(0,\infty)}F_{\nu^{*}}(x)\nu(dx)+\theta
F_{\nu^{*}}(0)F_{\nu
}(0)\leq\int_{(0,\infty)}F_{\nu^{*}}(x)
\nu^{*}(dx)+\theta F_{\nu
^{*}}(0)F_{\nu^{*}}(0).
\]
Hence, $(\nu^{*},\nu^{*})$ is a symmetric Nash equilibrium for the problem.
\end{pf*}

\begin{remark}
Substituting for the values of the optimal Lagrange multipliers, we find
\begin{eqnarray*}
V_{\nu^{*},\nu^{*}}^{1} & = & \int_{(0,\infty)}F_{\nu^{*}}(x)
\nu^{*}(dx)+\theta F_{\nu^{*}}(0)F_{\nu^{*}}(0)
\\
& = & 1-(1-\theta)F_{\nu^{*}}(0)F_{\mu}(0)-\frac{1}{2}
\bigl(1-F_{\nu
^{*}}(0)^{2} \bigr)
\\
& = & \frac{1}{2}+ \biggl(\theta-\frac{1}{2} \biggr)F_{\mu}(0)^{2}.
\end{eqnarray*}
Note that this is as expected, since for any law $\pi$ with no atom on
$(0,\infty)$ we have by symmetry that $V^1_{\pi,\pi} = (1 - \pi(\{
0\}
)^2)/2 + \theta\pi(\{0\})^2$.
\end{remark}

\section{Preliminaries} \label{secPreliminary}

In the previous section, we characterised the symmetric Nash equilibria.
In this section, we state and prove some auxiliary results which will
be required for proofs of existence and uniqueness in later sections.
The first result is of independent interest.
Note that $F_{\nu^*}$ is a concave function on $[0,\infty)$, but the
case we will want in the following theorem is for convex distribution functions.

\begin{proposition} \label{propBound}
Fix $y \in(0,\infty)$.
Let $\sP(y) \subseteq\sP$ be the set of probability measures $\pi$ on
$\R^+$
with mean $\overline{\pi} = y$. For $\pi\in\sP$, recall that
$F_\pi$
is the distribution
function of $\pi$, and extend this definition to $(-\infty,\infty)$.
Let $a_{\pi}=\inf\{u\dvtx F_\pi(u)>0\}\geq0$
and $b_\pi=\sup\{u\dvtx F_\pi(u)<1\}\leq\infty$.
\begin{enumerate}[(iii)]
\item[(i)] Let $\sP_{\mathrm{cx}}({y})= \{ \pi\in\sP({y})\dvtx F_\pi\mbox{ is convex
on } (-\infty,b_{\pi} ] \} $.
Then, for $\pi\in\sP_{\mathrm{cx}}(y)$, $\pi(\{0\})=0$. Moreover:
\begin{longlist}[(a)]
\item[(a)] suppose $H$ is convex; then,
\[
\sup_{\pi\in\sP_{\mathrm{cx}}(y)}\int\pi(dz)H(z)=\int_{0}^{2{y}}
\frac
{1}{2{y}}H(v)\,dv;
\]
\item[(b)] suppose $H$ is concave; then, $\sup_{\pi\in\sP_{\mathrm{cx}}(y)}\int
\pi (dz) H(z)=H({y})$.
\end{longlist}

\item[(ii)] Let $\sP_{\mathrm{cv}}({y})= \{ \pi\in\sP({y})\dvtx F_\pi\mbox{ is concave
on } [0,\infty) \} $.
Then, for $\pi\in\sP_{\mathrm{cv}}(y)$, $a_\pi=0$. Moreover:
\begin{longlist}[(a)]
\item[(a)] suppose $H$ is convex; then, $\sup_{\pi\in\sP_{\mathrm{cv}}(y)}\int
\pi(dz) H(z) = H(0) + y \lim_{x \uparrow\infty} \frac{H(x)}{x}$;

\item[(b)] suppose $H$ is concave; then
\[
\sup_{\pi\in\sP_{\mathrm{cv}}(y)}\int\pi(dz)H(z)=\int_{0}^{2{y}}
\frac
{1}{2{y}}H(v)\,dv.
\]
\end{longlist}
\end{enumerate}

The suprema in \textup{(i)(a)} and \textup{(ii)(b)} are
attained by $\pi\sim\mathcal{U}[0,2{y}]$. The supremum in \textup{(i)(b)} is
attained by $\pi\sim\delta_{{y}}$. The suprema in \textup{(i)(b)} and \textup{(ii)(a)}\vadjust{\goodbreak}
are valid for all distributions on $\mathbb{R}^{+}$ with mean ${y}$
and not just those with convex (or concave) distribution functions.
\end{proposition}

\begin{remark}
This result is stated for completeness; the result we will use and
prove is (i)(a).
\end{remark}

\begin{pf*}{Proof of Proposition \ref{propBound}}
Let $U$ be a $\mathcal{U}[0,1]$ random variable.

\begin{figure}[b]

\includegraphics{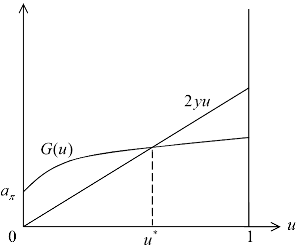}

\caption{Comparison of $G(u)$ and $2{y}u$. Since $G$ is the inverse
of the CDF of a random variable with mean $y$, the area under $G$
is ${y}$. Then the areas under $G$ and the line $\ell(u)=2{y}u$
are the same. Hence, if $G$ is concave on $[0,1]$ (and not a straight
line passing through the origin), there is
a unique crossing point $u^{*}$ of $G$ and the line $\ell$.}
\label{figBound}
\end{figure}

Suppose $\pi\in\sP_{\mathrm{cx}}(y)$ and let $Y$ have law $\pi$. It is obvious
that $F=F_\pi$, is strictly increasing on $(a_\pi,b_\pi)$,
and $F(a_\pi)=0$. Hence, the inverse function $G\triangleq F^{-1}$
exists on $[0,1]$. Since $G(U)$ is distributed as $Y$,
$\mathbb{E}[H(Y)]=\mathbb{E}[H(G(U))]=\int_{0}^{1}H(G(u))\,du$ and
$\int_{0}^{1}G(u)\,du=\mathbb{E}[G(U)]=\mathbb{E}[Y]={y}$.

It is clear that $G$ is concave on $[0,1]$ and $G(0)= a_\pi\geq0$.
Since $\int_{0}^{1}2{y}u\,du=\int_{0}^{1}G(u)\,du$, then either $G(u)=2 y
u$ and $Y\sim2 y U$, or there exists
a unique $u^{*}\in(0,1)$ such that $G(u^{*})=2{y}u^{*}$ (see
Figure~\ref{figBound}). In the latter case, $2{y}u\leq G(u)\leq2{y}u^{*}$
for $u\in[0,u^{*}]$ and $2{y}u^{*}\leq G(u)\leq2{y}u$ for
$u\in[u^{*},1]$. Then if $D \triangleq\mathbb{E}[H(Y)]-\mathbb
{E}[H(2{y}U)]$ we have
\[
D =\int_{0}^{u^{*}} \bigl(H \bigl(G(u)
\bigr)-H(2{y}u) \bigr)\,du+\int_{u^{*}}^{1} \bigl(H
\bigl(G(u) \bigr)-H(2{y}u) \bigr)\,du.
\]
Let $H_{-}'$ denote the left derivative of the convex function
$H$. Then $H(u_{2})-H(u_{1})\leq(u_{2}-u_{1})H_{-}'(u_{2})$ for
any $u_{1}$ and $u_{2}$, and setting $u_1 = 2 {y} u$ and $u_2 = G(u)$
and using the fact that
$H_-'$ is increasing,
\begin{eqnarray*}
D & \leq& \int_{0}^{u^{*}} \bigl(G(u)-2{y}u
\bigr)H_{-}' \bigl(G(u) \bigr)\,du+\int
_{u^{*}}^{1} \bigl(G(u)-2{y}u \bigr)H_{-}'
\bigl(G(u) \bigr)\,du
\nonumber
\\
& \leq& H_{-}' \bigl(2{y}u^{*} \bigr)\int
_{0}^{1} \bigl(G(u)-2{y}u \bigr)\,du=0.\label{eqBound2yZ}
\end{eqnarray*}
Thus, $\mathbb{E}[H(Y)]\leq\mathbb{E}[H(2{y}U)]=\int_{0}^{2{y}}\frac
{1}{2{y}}H(u)\,du$.
Moreover, it is obvious that the bound is attained by $Y\sim\mathcal
{U}[0,2{y}]$.
\end{pf*}

\begin{proposition} \label{propIncrDSlope}
Suppose that $H$ is such that $H$ is twice differentiable and $h=H'$ is concave.
Suppose that $H(0)=0$, $h(0)>0$, $h'(0)\leq0$ and $h$ is not
constant. Then, for $\hat{w}>0$ such that $H(\hat{w})=0$, we
have $h(\hat{w})+h(0)\leq0$, \textit{that is}, $\llvert h(\hat
{w})\rrvert\geq h(0)$.
\end{proposition}

\begin{pf}
Since $h$ is not constant and $h$ is concave, there is a solution
$w=\tilde{w}$ say, to $h(w)=-h(0)$. Let $\delta=-2h(0)/\tilde{w}$ be the
slope of the line joining $(0,h(0))$ to $(\tilde{w},-h(0))$. Then,
on $(0,\tilde{w})$, $h(w)\geq h(0)+\delta w$ and $H(w)=\int
_{0}^{w}h(x)\,dx\geq h(0)w+\delta w^{2}/2$,
so that $H(\tilde{w})\geq0$. Then, by concavity of $H$ and since
$H(\hat{w})=0$, $\hat{w}\geq\tilde{w}$. Thus, $h(\hat{w})\leq
h(\tilde
{w})=-h(0)$
and the result follows.
\end{pf}

\begin{proposition} \label{propSmooth}
For any measure $\nu\in\mathcal{A}_{\mu}^{*}$, if $C_{\nu
}(x)=C_{\mu}(x)$
for some \mbox{$x>0$}, then $C_{\mu}(\cdot)$ is differentiable at $x$
and $C_{\nu}'(x)=C_{\mu}'(x)$.
\end{proposition}

\begin{pf}
Suppose that $\nu\in\mathcal{A}_{\mu}^{*}$. Then $C_{\nu}$ is continuously
differentiable on $(0,\infty)$ and $C_{\nu}\geq C_{\mu}$. Since
$C_{\nu
}(y)-C_{\mu}(y)\geq0=C_{\nu}(x)-C_{\mu}(x)$
for any $y<x$, we deduce $C_{\nu}'(x-)-C_{\mu}'(x-)\leq0$.
Similarly, we have $C_{\nu}'(x+)-C_{\mu}'(x+)\geq0$. Thus,
$C_{\mu}'(x+)\leq C_{\nu}'(x+)=C_{\nu}'(x-)\leq C_{\mu}'(x-)$.
Conversely,\vspace*{1pt} convexity of $C_{\mu}$ implies $C_{\mu}'(x-)\leq C_{\mu
}'(x+)$, and hence $C_{\mu}'(x-)=C_{\mu}'(x+)$, and the results follow.
\end{pf}

\begin{proposition} \label{propQuadBound}
Fix any measure $\nu\in\mathcal{A}_{\mu}^{*}$. Suppose that $C_{\nu
}(x)=\phi(x)$
on some interval $\mathcal{J}=[j_{1},j_{2})$, where $0\leq j_{1}<j_{2}
\leq\infty$
and where $\phi(\cdot)$ is a quadratic function defined on $(-\infty
,\infty)$
with $\phi''>0$. Then $j_{2}<\infty$ and $C_{\nu}(x)\leq\phi(x)$
on $[j_{2},\infty)$.
\end{proposition}

\begin{pf}
Assume that $j_{2}=\infty$. Then $C_{\nu}$ is quadratic on
$[j_{1},\infty)$
with strictly positive quadratic coefficient. This means that $C_{\nu}$
is not ultimately decreasing, which is a contradiction. Thus,
$j_{2}<\infty$.
Since $C_{\nu}'$ is continuous and concave on $(0,\infty)$ and since
$\phi{'}$
is linear, it is clear that $C_{\nu}(x)\leq\phi(x)$ on
$[j_{2},\infty)$.
\end{pf}

\section{Existence of a Nash equilibrium}
\label{secExistence}

\subsection{Atomic initial measure}\label{secAtomicMeas}

We start with the case where the initial law $\mu$ is an atomic probability
measure. We will construct a put function $Q(x)$ that satisfies certain
conditions, and then define a measure $\nu$ via $\nu((-\infty,x])= Q'(x)$.
It will then follow that $\nu$ belongs to $\mathcal{A}_{\mu}^{*}$.

We state the theorem for the case of a measure $\chi$ as we will need
the more general result in subsequent sections. In the case of $X_0
\sim\mu$, where $\mu$ is a purely atomic probability measure with
finitely many atoms, the theorem gives existence of a symmetric Nash
equilibrium.

\begin{theorem} \label{thmAtomicMeas}
Suppose $\chi\in\sM$ consists of finitely many atoms, that is, $\chi
=\sum_{j=1}^{N}p_{j}\delta_{\xi_{j}}$,
where $0\leq\xi_{1}<\xi_{2}<\cdots<\xi_{N}$ and $p_{j} > 0$
for all $1\leq j\leq N$.\vspace*{2pt} Then $\sA^*_\chi$ is nonempty.
\end{theorem}

\begin{pf}
If $\chi$ is a point mass at zero, then set $\pi(\{0\})=\chi(\{0\})$.
Then $\pi\in\sA^*_\chi$
and the construction is complete.

Otherwise, set
$Q_{1}(r,y)=yF_{\chi}(0)+ry^{2}/2$.
Then there exists a unique value of $r$ ($r_{1}$ say) such that
\begin{eqnarray*}
Q_{1}(r_{1},y)&\geq& P_{\chi}(y)\qquad \forall y\geq0
\quad\mbox{and}
\\
Q_{1}(r_{1},y)&=&P_{\chi}(y) \qquad\mbox{for some }y>0.
\end{eqnarray*}
Let $y_{1}=\max\{y>0\dvtx Q_{1}(r_{1},y)=P_{\chi}(y)\}$. That $y_1$ is
finite is one of the conclusions of Proposition~\ref{propQuadBound},
but also follows from the fact that $Q_1(r_1, \cdot)$ is quadratic,
whereas $P_\chi(\cdot)$ is ultimately linear.
Then necessarily
$P_{\chi}'(y_{1})$ exists and $\frac{\partial}{\partial
y}Q_{1}(r_{1},y_{1})=P_{\chi}'(y_{1})$ (see the proof of
Proposition~\ref{propSmooth}).
Note that $y_{1}\notin\{\xi_{1},\xi_{2},\ldots,\xi_{N}\}$ since
$P_{\chi}'$
has a kink at these points. Set $\xi_{0}=0$ and $\xi_{N+1}=\infty$ and
let $n_{1}$ be such that $\xi_{n_{1}}<y_{1}<\xi_{n_{1}+1}$.
If $n_{1}=N$ [equivalently $P_{\chi}'(y_{1})=\sum_{j=1}^{N}p_{j}=\chi
(\mathbb{R}^+)$]
then stop. Otherwise, we proceed inductively.

\begin{figure}

\includegraphics{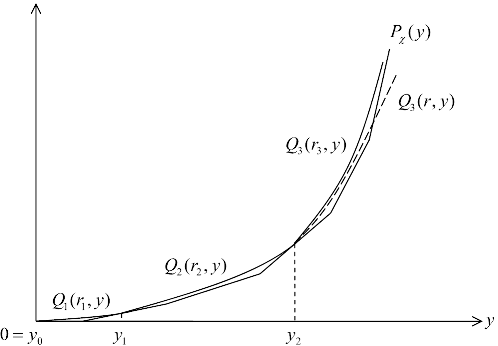}

\caption{The construction of $Q(y)$. The piecewise linear curve is
$P_{\chi}(y)$. The functions
$Q_{3}(\cdot,y)$ are quadratic functions of $y$ on $[y_2,\infty)$ such
that $Q_{3}(\cdot,y_{2})=P_{\chi}(y_{2})$
and $\frac{\partial}{\partial y}Q_{3}(\cdot,y_{2})=P_{\chi}'(y_{2})$.
Then $r_{3}$ is the unique value of $r$ such that $Q_{3}(r_{3},y)\geq
P_{\chi}(y)$
for all $y\geq y_{2}$ and $Q_{3}(r_{3},y)=P_{\chi}(y)$ for some
$y>y_{2}$. The dashed curve is $Q_{3}(r,y)$ with $r<r_{3}$.}
\label{figQConst}
\end{figure}

Let $y_{0}=0$. Suppose we have found $0<y_{1}<y_{2}<\cdots<y_{k}<\xi_{N}$
($y_{i}\notin\{\xi_{1},\xi_{2},\ldots,\xi_{N}\}$ $\forall1\leq
i\leq k$)
and $Q(\cdot)$ on $[0,y_{k}]$ such that (i) $Q$ is continuously differentiable,
(ii) $Q(y_{i})=P_{\chi}(y_{i})$ and $Q'(y_{i})=P_{\chi}'(y_{i})$
for any $1\leq i < k$,
(iii)~$Q(y_{k})=P_{\chi}(y_{k})$ and $Q'(y_{k}-)=P_{\chi}'(y_{k})$ and
(iv)~$Q''$ is defined everywhere except at the points $0,y_1, y_2,\ldots
, y_k$ and is piecewise constant and decreasing. In particular,
$Q$ is quadratic on $ \{ (y_{i-1},y_{i}) \} _{1\leq i\leq k}$
with representation $Q(y)=Q_{i}(r_{i},y)$ for $y \in[y_{i-1},y_i]$ where
\[
Q_{i}(r_{i},y)\triangleq P_{\chi}(y_{i-1})+(y-y_{i-1})P_{\chi
}'(y_{i-1})+
\tfrac{1}{2}r_{i}(y-y_{i-1})^{2},
\]
and where $(r_{i})_{1\leq i\leq k}$
is a strictly decreasing sequence. Let $Q_{k+1}(r,y)=P_{\chi
}(y_{k})+(y-y_{k})P_{\chi}'(y_{k})+\frac{1}{2}r(y-y_{k})^{2}$;
then there exists a unique $r$ ($r_{k+1}$ say) such that
\[
Q_{k+1}(r_{k+1},y)\geq P_{\chi}(y)\qquad \forall y\geq
y_{k}
\]
and
\[
Q_{k+1}(r_{k+1},y)=P_{\chi}(y)\qquad\mbox{for some
}y>y_{k}.
\]
See Figure \ref{figQConst}. Since $Q_{k}(r_{k},y)>P_{\chi}(y)$ for all $y>y_{k}$, it is clear
that $0<r_{k+1}<r_{k}$. Set $y_{k+1}=\max\{
y>y_{k}\dvtx Q_{k+1}(r_{k+1},y)=P_{\chi}(y)\}$.
Then $P_{\chi}'(y_{k+1})$ exists, $P_{\chi}'(y_{k+1})=\frac{\partial
}{\partial y}Q_{k+1}(r_{k+1},y_{k+1})$,
and $y_{k+1}\notin\{\xi_{1},\xi_{2},\ldots,\xi_{N}\}$ since
$P_{\chi}$
has changes in slope at these points. Set $Q(y)=Q_{k+1}(r_{k+1},y)$
on $[y_{k},y_{k+1}]$.

We repeat the construction up to and including the index $k=T-1$
for which $y_{k+1}>\xi_{N}$. Then $y_{T-1}<\xi_{N}<y_{T}$. Finally, we
set $Q(y)=P_{\chi}(y) = \chi(\R^+) y - \overline{\chi}$ for $y\geq y_{T}$.

For $y>0$, let $\rho(y)=Q''(y)$. Then $\rho$ is defined almost
everywhere and $\rho(y)=r_{i}$ on $(y_{i-1},y_{i})_{1\leq i\leq T}$
and $\rho(y)=0$ on $(y_{T},\infty)$. Furthermore, $\rho$ is decreasing
and $\rho$ only decreases at points where $P_{\chi}(y)=Q(y)$. Let
$\pi$ be the measure with an atom at 0 of size
$F_{\chi}(0)$
and
density $\rho$ on $(0,\infty)$, and recall that $y_{T}>\xi_{N}$. Then
$F_{\pi}(0)=F_{\chi}(0)$; for any $y\geq y_{T}$, $F_{\pi}(y)=P_{\pi
}'(y_{T})=P_{\chi}'(y_{T})=\chi(\mathbb{R}^+)$
and $\overline{\pi} = \int_{0}^{\infty}y\pi(dy)=y_{T}F_{\pi
}(y_{T})-P_{\pi}(y_{T})=y_{T}F_{\chi}(y_T)-P_{\chi}(y_{T})=\overline
{\chi}$.
Furthermore, $P_{\pi}(y)=Q(y) \geq P_{\chi}(y)$.
It follows that $\pi\in\sA^*_{\chi}$.
\end{pf}

\begin{remark} \label{rmkAtomicMap}
Fix any $k\in\{1,2,\ldots,T\}$. Let $n_{k}$ be such that $\xi
_{n_{k}}<y_{k}<\xi_{n_{k}+1}$.
Then, in the mapping $\chi\mapsto\pi$, the atoms\vspace*{2pt} $(\xi_{1},\xi
_{2},\ldots,\xi_{n_{k}})$
of $\chi$ are mapped to $[0,y_{k}]$, and $\pi([0,y_{k}])=P_{\chi
}'(y_{k})=\sum_{j=1}^{n_{k}}p_{j}$.
Moreover, $\int_{0}^{y_{k}}y\pi(dy)=y_{k}F_{\pi}(y_{k})-P_{\pi
}(y_{k})=y_{k}\sum_{j=1}^{n_{k}}p_{j}-P_{\chi}(y_{k})=y_{k}\sum
_{j=1}^{n_{k}}p_{j}-\sum_{j=1}^{n_{k}}p_{j}(y_{k}-\xi_{j})=\sum
_{j=1}^{n_{k}}p_{j}\xi_{j}$.
\end{remark}

\begin{example}
Suppose that $\chi=p\delta_{\xi}$ with $\xi>0$. Then $P_{\chi
}(y)=p(y-\xi)^{+}$.
Let $Q_{1}(r,y)=ry^{2}/2$. Then, $Q_{1}(r,y)\geq P_{\chi}(y)$ if
and only if $r\geq p/(2\xi)\triangleq r_{1}$, and for $r=r_{1}$ we have
$Q_{1}(r_{1},y)\geq P_{\chi}(y)$
with equality at $y=0$ and $y=y_{1}=2\xi$. Then $y_{1}>\xi$
so that the construction ends and $\pi=p \mathcal{U}[0,2\xi]$.
\end{example}

The rest of this subsection is devoted to the proof of a useful
proposition which will be used
in the next subsection to find the optimal target law for a general initial
measure.

For $\varpi\in\sM$, let $X_\varpi$ be a random variable with law
$\varpi$. Denote by $\breve{\varpi}$ the law of a random variable
$\breve{X}_\varpi$
where conditional on $X_\varpi= x$, $\breve{X}_\varpi$ has law $ \uU_x
= \mathcal{U}[0,2x]$.

\begin{lemma}
\label{lembreve}
For $x \geq0$, and $\varpi\in\sM$,
%
\begin{eqnarray}\label{eqnPbreve}
F_{\breve{\varpi}}(x) &=& F_\varpi(x/2) + \int
_{(x/2,\infty)} \varpi(dy) \frac{x}{2y};
\nonumber\\[-8pt]\\[-8pt]\nonumber
P_{\breve{\varpi}}(x) &=& \int_{x/2}^\infty
P_\varpi(u) \frac{x^2}{2 u^3} \,du.
\end{eqnarray}
\end{lemma}

\begin{pf}
We prove the second result. We have
\begin{eqnarray*}
P_{\breve{\varpi}}(x) & = & xF_{\varpi}(0) + \int_{(0,\infty
)}
\varpi(du)\int_{0}^{2u}\frac{(x-z)^{+}}{2u}\,dz
\\
& = & xF_{\varpi}(0) + \int_{(0,x/2)}(x-u)\varpi(du)+\int
_{[{x}/{2},\infty)}\frac{x^{2}}{4u}\varpi(du),
\end{eqnarray*}
and integrating by parts we find
\[
P_{\breve{\varpi}}(x) = \int_{(0,x/2)}F_{\varpi}(u)\,du +
\int_{[x/2,\infty)}F_{\varpi}(u)\frac{x^{2}}{4u^{2}}\,du.
\]
The result follows from a further integration by parts.
\end{pf}

\begin{corollary}
\label{corbreve}
If $\pi\preceq_{\mathrm{cx}} \varpi$, then $\breve{\pi} \preceq_{\mathrm{cx}}
\breve
{\varpi}$.
\end{corollary}

\begin{proposition} \label{propProperties}
Suppose $\mu\in\sM$ consists of finitely many atoms.
Denote by $\nu$ the element of $\sA^*_\mu$
which is constructed using the algorithm in Theorem \ref
{thmAtomicMeas}. Suppose $\varpi$ is any
measure such that $\varpi$ has mass $\mu(\R^+)$, mean $\overline
{\mu}$
and $\mu\preceq_{\mathrm{cx}}\varpi$. Then
$\nu\preceq_{\mathrm{cx}}\breve{\varpi}$.
\end{proposition}

\begin{pf}
By Corollary~\ref{corbreve}, it is sufficient to prove the proposition
in the case $\varpi= {\mu}$.

Suppose $\mu=\sum_{j=1}^{N}p_{j}\delta_{\xi_{j}}$, where $0\leq\xi
_{1}<\xi_{2}<\cdots<\xi_{N}$,
$p_{i}>0$ and\break $\sum_{j=1}^{N}p_{j}\xi_{j}=\overline{\mu}$.
By construction, $\breve{\mu}=\sum_{j=1}^{N}p_{j}\mathcal{U}[0,2\xi
_{j}]$ where
$\mathcal{U}[0,0]$ can more simply be written as $\delta_{0}$.

For any $m\in\{1,2,\ldots,N\}$, define $\mu^{m}=\sum
_{j=1}^{m}p_{j}\delta_{\xi_{j}}$
and suppose $\nu^{m}$ is the corresponding measure derived using
the algorithm in Theorem \ref{thmAtomicMeas}.

If $N=1$, then $\mu=p_{1}\delta_{\xi_{1}}$ and $\nu=p_{1}\mathcal
{U}[0,2\xi_{1}] = \breve{\mu}$
and the result holds.

Now suppose that $N\geq2$. Then $\nu=\nu^{N}=\sum_{m=1}^{N-1}(\nu
^{m+1}-\nu^{m})+\nu^{1}$.
Note that $\mu^{m+1} - \mu^{m} \sim p_{m+1} \delta_{\xi_{m+1}}$ and
hence $\nu^{m+1} - \nu^{m}$ has mass $p_{m+1}$ and mean $\xi_{m+1}$.
Provided we can show that $(\nu^{m+1}-\nu^{m})$ has a nondecreasing density,
then it follows from Proposition \ref{propBound} that $(\nu
^{m+1}-\nu
^{m})\preceq_{\mathrm{cx}}p_{m+1}\mathcal{U}[0,2\xi_{m+1}]$.
Then, since convex order is preserved under addition, if $(\nu
^{m+1}-\nu
^{m})$ has a nondecreasing density for every $m$ with $1 \leq m \leq
N-1$, then $\nu\preceq_{\mathrm{cx}}\sum_{m=1}^{N}p_{m}\mathcal{U}[0,2\xi
_{m}]=\breve{\mu}$.

We use a suffix $m$ to label quantities constructed in Theorem \ref
{thmAtomicMeas},
to show that they are constructed from measure $\mu^{m}$. The idea of
the proof is that in calculating $\nu^m$ and $\nu^{m+1}$ using the
algorithm of Theorem~\ref{thmAtomicMeas}, the early parts of the
construction will be the same, and indeed $\nu^m$ and $\nu^{m+1}$ will
differ only over the final nonzero element of $\nu^{m+1}$.

Fix $m$ with $1 \leq m \leq N$.
Define $\sB^m \subseteq\{1, 2, \ldots, T^m \}$ by $\sB^m = \{ k\dvtx\break
Q_{k}^{m}(r_{k}^{m},y)\geq P_{\mu^{m+1}}(y) \mbox{ on
}(y_{k-1}^{m},\infty) \}$.

\begin{longlist}[\textit{Case} (a)]
\item[\textit{Case} (a).] $\sB^m = \{1, 2, \ldots,T^m \}$.

Then\vspace*{1pt} $ (Q_{j}^{m}(r_{j}^{m},y),y_{j}^{m} )_{1 \leq j \leq T^m}$
and $ (Q_{j}^{m+1}(r_{j}^{m+1},y),y_{j}^{m+1} )_{1 \leq j
\leq T^m}$
are the same. Then also $T^{m+1} = T^m + 1$,
$y_{T^{m}}^{m}<y_{T^{m+1}}^{m+1}$ and the densities $\rho^{m+1}$ and
$\rho^{m}$ satisfy that $\rho^{m+1}=\rho^{m}$ on the interval
$(0,y_{T^{m}}^{m}=y_{T^{m}}^{m+1})$,
$\rho^{m+1}$ is\vspace*{1pt} constant on $(y_{T^{m}}^{m+1},y_{T^{m+1}}^{m+1})$
and $\rho^{m}$ is zero on $(y_{T^{m}}^{m+1},y_{T^{m+1}}^{m+1})$.
In particular, $\nu^{m+1}- \nu^{m} = p_{m+1} \sU[y^{m+1}_{T^m},
y^{m+1}_{T^{m}}] \preceq_{\mathrm{cx}} p_{m+1} \sU[0, 2 \xi_{m+1}]$.



\begin{figure}

\includegraphics{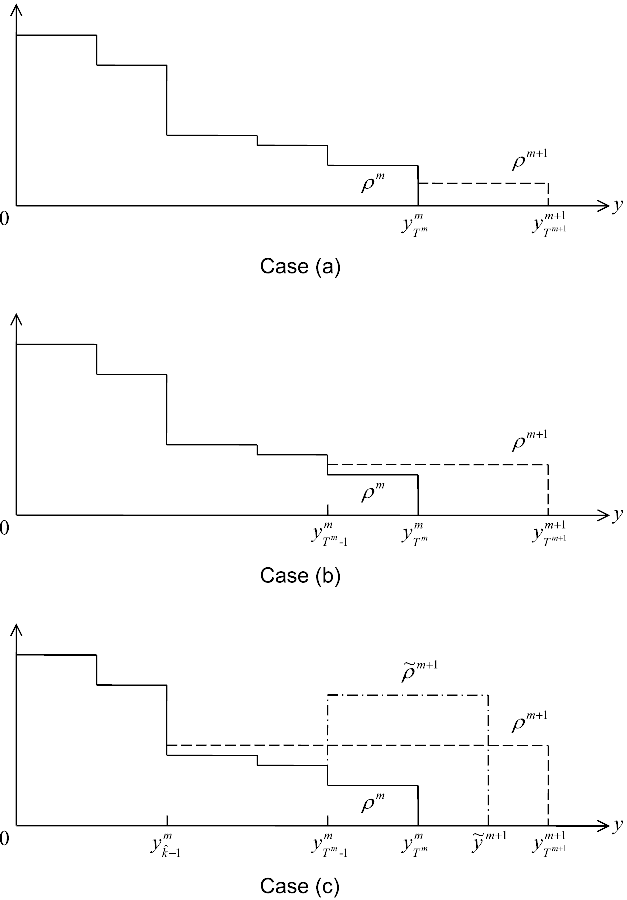}

\caption{Graph of the decreasing, piecewise constant functions $\rho
^{m}(y)$ and $\rho^{m+1}(y)$.
Observe that $\rho^{m}\leq\rho^{m+1}$ and $(\rho^{m+1}-\rho^{m})$
is nondecreasing on $(0,y_{T^{m+1}}^{m+1})$. The three cases
correspond to the three cases in the proof. In case~\textup{(c)} the density
$\tilde{\rho}^{m+1}$ is also shown.}
\label{figDensities}
\end{figure}

\item[\textit{Case} (b).]  $\inf\{ k\dvtx k \notin\sB^m \} = T^m$.

Then it must be that in the construction we have $T^m = T^{m+1}$, $\rho
^{m+1}=\rho^{m}$ on the interval $(0,y^m_{T^m - 1} \equiv y_{T^m - 1}^{m+1})$, $\rho^{m+1}$
is constant (with value denoted by $r_{T^{m+1}}^{m+1}$ say) on $(y_{T^m
-1}^{m+1},y_{T^{m+1}}^{m+1})$,
and $\rho^{m}$ is constant and strictly less than $r_{T^{m+1}}^{m+1}$
on $(y_{T^m-1}^{m+1},y_{T^{m}}^{m})$, $\rho^{m}$ is zero on
$(y_{T^{m}}^{m},\infty)$. We want to
argue that $y_{T^{m}}^{m}<y_{T^m}^{m+1}$, which then implies
that $(\rho^{m+1}-\rho^{m})$ is nondecreasing on
$(0,y_{T^{m}}^{m+1})$. See case~(b) of Figure~\ref{figDensities}.

Note that in the construction of $\nu^{m}$ the masses at points $(\xi
_{n^m_{T^{m}-1}+1},\ldots,\xi_{m})$
are embedded in the interval $(y_{T^{m}-1}^{m},y_{T^{m}}^{m})$, and
in the construction of ${\nu}^{m+1}$ the masses at points $(\xi
_{n^m_{T^{m}-1}+1},\ldots,\xi_{m+1})$
are embedded in $(y_{T^{m}-1}^{m},{y}_{T^m}^{m+1})$. Moreover, $\nu^{m}$
has constant density over $(y_{T^{m}-1}^{m},y_{T^{m}}^{m})$ and ${\nu}^{m+1}$
has constant density over $(y_{T^{m}-1}^{m},{y}_{T^m}^{m+1})$. Consider
the means of $\nu^{m}$ and ${\nu}^{m+1}$; by Remark~\ref
{rmkAtomicMap}, we have
\[
\frac{1}{2} \bigl(y_{T^{m}-1}^{m}+y_{T^{m}}^{m}
\bigr)\sum_{j=n^m_{T^{m}-1}+1}^{m}p_{j}=\sum
_{j=n^m_{T^{m}-1}+1}^{m}p_{j}
\xi_{j}
\]
and
\[
\frac{1}{2} \bigl(y_{T^{m}-1}^{m}+{y}_{T^m}^{m+1}
\bigr)\sum_{j=n^m_{T^{m}-1}+1}^{m+1}p_{j}=\sum
_{j=n^m_{T^{m}-1}+1}^{m+1}p_{j}
\xi_{j}.
\]
Hence,
\begin{eqnarray*}
y_{T^{m}}^{m}-y_{T^{m}-1}^{m} & =&
\frac{2\sum_{j=n^m_{T^{m}-1}+1}^{m}p_{j} (\xi_{j}-y_{T^{m}-1}^{m}
)}{\sum_{j=n^m_{T^{m}-1}+1}^{m}p_{j}}
\\
& <&\frac{2\sum_{j=n^m_{T^{m}-1}+1}^{m+1}p_{j} (\xi
_{j}-y^m_{T^{m}-1} )}{\sum
_{j=n^m_{T^{m}-1}+1}^{m+1}p_{j}}={y}_{T^m}^{m+1}-y_{T^{m}-1}^{m}
\end{eqnarray*}
and then $y_{T^{m}}^{m}<{y}_{T^m}^{m+1}$.

\item[\textit{Case} (c).] $\inf\{ k\dvtx k \notin\sB^m \} < T^m$.\vspace*{1pt}

Define $\hat{k}\triangleq\inf\{ k\dvtx k \notin\sB^m \}$.
Then it must be that in the construction we have $T^{m+1} =\hat{k}$,
$\rho^{m+1}=\rho^{m}$
on the interval $(0,y_{\hat{k}-1}^{m} \equiv y_{\hat{k}-1}^{m+1})$,
$\rho^{m+1}$
is constant (with value denoted by $r_{T^{m+1}}^{m+1}$ say) on
$(y_{\hat
{k}-1}^{m+1},y_{T^{m+1}}^{m+1})$,
$\rho^{m}$ is decreasing and strictly less than $r_{T^{m+1}}^{m+1}$
on $(y_{\hat{k}-1}^{m+1},y_{T^{m}}^{m})$ and $\rho^{m}$ is zero on
$(y_{T^{m}}^{m},\infty)$. Similarly to case~(b), we want to
argue that $y_{T^{m}}^{m}<y_{T^{m+1}}^{m+1}$, which then implies
that $(\rho^{m+1}-\rho^{m})$ is zero on $(0,y^m_{\hat{k}-1})$ and
nondecreasing on $(y^m_{\hat{k}-1},y_{T^{m+1}}^{m+1})$. See case~(c)
of Figure~\ref{figDensities}.

We first construct a new measure $\tilde{\nu}^{m+1}$.
Define $\tilde{Q}^{m+1}$ on $[0, y^m_{T^m-1}]$ by $\tilde{Q}^{m+1}(y) =
Q^m(y) = P_{\nu^m}(y)$.
Let $L^{m+1}$ be the line $L^{m+1}(y)=\sum_{j=1}^{m+1}p_{j}(y-\xi_{j})$
so that $P_{\mu^{m+1}}(y) = \max\{ P_{\mu^m}(y), L^{m+1}(y) \}$ and
for $y \geq y^m_{T^m-1}$ define
\[
\tilde{Q}_{T^{m}}^{m+1}(r,y)\triangleq P_{\mu
^{m}}
\bigl(y_{T^{m}-1}^{m} \bigr)+ \bigl(y-y_{T^{m}-1}^{m}
\bigr)P_{\mu
^{m}}' \bigl(y_{T^{m}-1}^{m}
\bigr)+\tfrac{1}{2}r \bigl(y-y_{T^{m}-1}^{m}
\bigr)^{2}.
\]
Then there exists a unique $r$ (denoted by $\tilde{r}^{m+1}$ say) such that
\[
\tilde{Q}_{T^{m}}^{m+1} \bigl(\tilde{r}^{m+1},y \bigr)
\geq L_{m+1}(y)\qquad \forall y\geq y_{T^{m}-1}^{m}
\]
and
\[
\tilde{Q}_{T^{m}}^{m+1} \bigl(\tilde{r}^{m+1},y
\bigr)=L_{m+1}(y)\qquad\mbox{for some }y>y_{T^{m}-1}^{m}.
\]
Note that in the construction of $\tilde{\nu}^{m+1}$ (unlike in the
construction of $\nu^m$ or ${\nu}^{m+1}$) there is no requirement that
$\tilde{r}^{m+1} \leq r^m_{T^m - 1}$.
Let $\tilde{y}^{m+1}$ be the point such that $\tilde
{Q}_{T^{m}}^{m+1}(\tilde{r}^{m+1},\tilde{y}^{m+1})=L_{m+1}(\tilde{y}^{m+1})$.
Then $\tilde{y}^{m+1}>y_{T^{m}-1}^{m}$ and $\frac{\partial}{\partial
y}\tilde{Q}_{T^{m}}^{m+1}(\tilde{r}^{m+1},\tilde
{y}^{m+1})=L_{m+1}'(\tilde{y}^{m+1}) = \sum_{j=1}^{m+1} p_j$.
Now let $\tilde{Q}^{m+1}(\cdot)$ be given by (see Figure~\ref{figNewMeas})
\begin{eqnarray*}
\tilde{Q}^{m+1}(y) & =&P_{\nu^{m}}(y)\cdot\mathbf
{1}_{[0,y_{T^{m}-1}^{m})}+\tilde{Q}_{T^{m}}^{m+1} \bigl(\tilde
{r}^{m+1},y \bigr)\cdot\mathbf{1}_{[y_{T^{m}-1}^{m},\tilde{y}^{m+1})}
\\
&&{} +L_{m+1}(y)\cdot\mathbf{1}_{[\tilde{y}^{m+1},\infty)}.
\end{eqnarray*}

Let $\tilde{\rho}^{m+1}=(\tilde{Q}^{m+1})''$, and let $\tilde{\nu}^{m+1}$
be the measure with density $\tilde{\rho}^{m+1}$ on $(0,\infty)$ and an
atom at 0
of size $F_{\mu}(0)$. Then $P_{\tilde{\nu}^{m+1}}(y)=\tilde{Q}^{m+1}(y)$,
and for $y\geq\tilde{y}^{m+1}$ we have $F_{\tilde{\nu
}^{m+1}}(y)=L_{m+1}'(\tilde{y}^{m+1})=\sum_{j=1}^{m+1}p_{j}$
and $\int_{0}^{\infty}y\tilde{\nu}^{m+1}(dy)=\tilde
{y}^{m+1}F_{\tilde
{\nu}^{m+1}}(\tilde{y}^{m+1})-P_{\tilde{\nu}^{m+1}}(\tilde{y}^{m+1})
=\sum_{j=1}^{m+1}p_{j}\xi_{j}$. In particular, $\tilde{\nu}^{m+1}$ has
the same mass and first moment as $\mu^{m+1}$, and hence as $\nu^{m+1}$.

\begin{figure}

\includegraphics{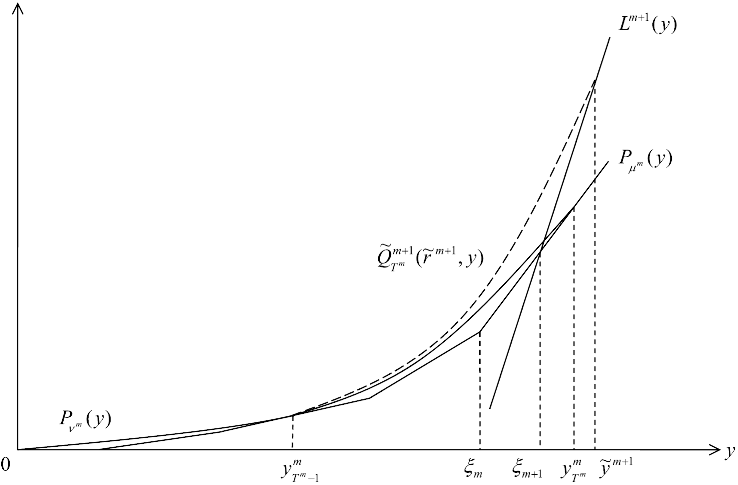}

\caption{Graph of $\tilde{Q}^{m+1}(y)$. The dashed curve $\tilde
{Q}_{T^{m}}^{m+1}(\tilde{r}^{m+1},y)$
is a quadratic function of $y$ over the interval $[y^m_{T^m - 1},
\tilde
{y}^{m+1}]$.}\vspace*{-4pt}\label{figNewMeas}
\end{figure}

The point about the intermediate measure $\tilde{\nu}^{m+1}$ is that as
in case~(b) the masses
at points $(\xi_{n^m_{T^{m}-1}+1},\ldots,\xi_{m+1})$
are embedded in the interval $(y_{T^{m}-1}^{m},\tilde{y}^{m+1})$.
In particular,
$\nu^{m}$
has constant density over $(y_{T^{m}-1}^{m},y_{T^{m}}^{m})$ and $\tilde
{\nu}^{m+1}$
has constant density over $(y_{T^{m}-1}^{m},\tilde{y}^{m+1})$. Then,
exactly as in the proof of case~(b),
considering
the means of $\nu^{m}$ and $\tilde{\nu}^{m+1}$, we have
$y_{T^{m}}^{m}<\tilde{y}^{m+1}$.

\begin{figure}

\includegraphics{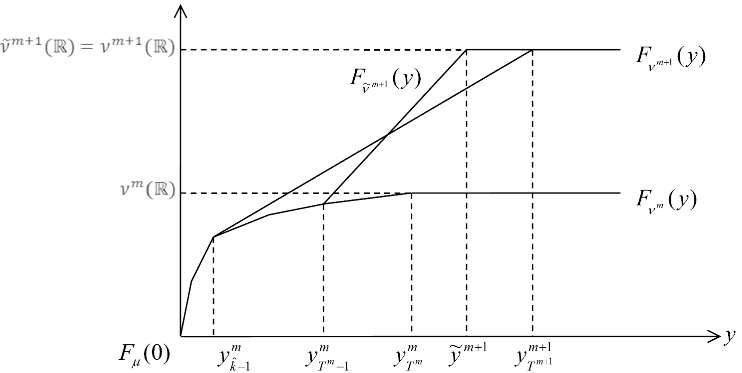}

\caption{Graph\vspace*{1pt} of $F_{\nu^{m+1}}$, $F_{\nu^{m}}$ and $F_{\tilde{\nu
}^{m+1}}$
in the case $\hat{k}<T^{m}$. By the constructions, $F_{\nu
^{m+1}}(y)=F_{\nu^{m}}(y)=F_{\tilde{\nu}^{m+1}}(y)$
on $[0,y_{\hat{k}-1}^{m+1}]$, $F_{\nu^{m+1}}(y)>F_{\nu
^{m}}(y)=F_{\tilde
{\nu}^{m+1}}(y)$
on $(y_{\hat{k}-1}^{m+1},y_{T^{m}-1}^{m}]$ and $F_{\tilde{\nu}^{m+1}}$
is linear on $(y_{T^{m}-1}^{m},\tilde{y}^{m+1})$. Since $F_{\nu^{m+1}}$
and $F_{\tilde{\nu}^{m+1}}$ have the same mean, the area between
the horizontal line at $\tilde{\nu}^{m+1}(\mathbb{R})$ and $F_{\nu
^{m+1}}$ must
be equal to the area between the line at $\tilde{\nu}^{m+1}(\mathbb{R})$
and $F_{\tilde{\nu}^{m+1}}$. Hence, $\tilde
{y}^{m+1}<y_{T^{m+1}}^{m+1}$.}\label{figDistrCompa}
\end{figure}

Next, we wish to compare the supports of $\tilde{\nu}^{m+1}$ and $\nu
^{m+1}$. Recall that $\tilde{\nu}^{m+1}(\mathbb{R}^+)=\nu
^{m+1}(\mathbb
{R}^+) = \sum_{j=1}^{m+1} p_j$
and $\tilde{\nu}^{m+1}$ and $\nu^{m+1}$ have the same mean. Moreover,
$F_{\nu^{m+1}}(y)=F_{\nu^{m}}(y)=F_{\tilde{\nu}^{m+1}}(y)$ on
$[0,y_{\hat{k}-1}^{m}]$,
and $F_{\nu^{m+1}}(y)>\break F_{\nu^{m}}(y)=F_{\tilde{\nu}^{m+1}}(y)$
on $(y_{\hat{k}-1}^{m},y_{T^{m}-1}^{m}]$. This implies that $\tilde
{y}^{m+1}<y_{T^{m+1}}^{m+1}$ since the means of $\tilde{\nu}^{m+1}$ and
$\nu^{m+1}$ are the same [and hence the area between $F_{\nu^{m+1}}$
and the horizontal line at height $\nu^{m+1}(\R^+)$ is equal to the
area between
$F_{\tilde{\nu}^{m+1}}$ and the same horizontal line (see Figure~\ref
{figDistrCompa})]. Hence,
$y^m_{T^m} < \tilde{y}^{m+1} < y^{m+1}_{T^{m+1}}$
and we find that $(\rho^{m+1}-\rho^{m})$ is nondecreasing
on $(0,y_{T^{m+1}}^{m+1})$, and $(\nu^{m+1}-\nu^{m})$ is a positive
measure with increasing density on its support. Hence, $\nu^{m+1} -
\nu
^m \preceq_{\mathrm{cx}} p_{m+1} \sU[0, 2 \xi_{m+1}]$.\quad\qed
\end{longlist}\noqed
\end{pf}

\subsection{General initial measure}\label{secGenerMeas}

Our first result shows that if $\mu$ is a general measure in $\sM$ then
$\sA^*_\mu$ is nonempty, and hence a symmetric Nash equilibrium exists.

\begin{theorem} \label{thmGeneralMeas}
Suppose $\mu\in\sM$. Then $\sA^*_\mu$ is nonempty.
\end{theorem}

\begin{pf}
Let $ \{ \mu_{n} \} _{n\geq1}$ be a sequence of atomic probability
measures with finite support such that $F_{\mu_{n}}(0)=F_{\mu}(0)$,
$\mu_{n}$ has total mass $\mu(\R^+)$ and mean $\overline{\mu}$ and
$\mu
_{n}\uparrow\mu$ in convex order.
Theorem~\ref{thmAtomicMeas} implies that for every ${n}$,
there exists
$\nu_{n}\in\mathcal{A}_{\mu_{n}}^{*}$. Define $D_n\dvtx [0,\infty)
\mapsto
[0,\mu(\R^+)]$ by
\[
D_{n}(x)=-C_{\nu_{n}}'(x)=\mu\bigl(\R^+
\bigr)-F_{\nu_{n}}(x)
\]
and let $b_{n}=\sup\{ x\dvtx F_{\nu_{n}}(x)< \mu(\R^+) \}
$. From
the construction
of $\nu_{n}$ in the proof of Theorem \ref{thmAtomicMeas} it can be
seen that $b_{n}$ is finite. Thus,
$D_{n}$ is a decreasing, convex function with $D_{n}(0)=\mu(\R
^+)-F_{\mu}(0)$,
$D_{n}(b_{n})=0$, $D_{n}\geq0$ and
%
\begin{equation}
\int_{0}^{\infty}D_{n}(x)\,dx=\int
_{0}^{b_{n}}D_{n}(x)\,dx=\int
_{0}^{b_{n}}x\nu_{n}(dx)=\overline{
\nu}_n = \overline{\mu}_n = \overline{
\mu}.\label{eqIntegralDn}
\end{equation}

Helly's Theorem [see \citet{Helly1912} or
\citet{Filipow2012}, Theorem 1.3]  states that
if $\{f_{n}\}_{n\geq1}$ is a uniformly bounded sequence of
monotone real-valued functions defined on $\mathbb{R}$ then there
is a subsequence $\{f_{n_{k}}\}_{k\geq1}$ which is pointwise convergent.
Hence, there exists a convergent subsequence of $\{D_{n}\}_{n\geq1}$ and
moving to a subsequence as necessary, we may assume $\{D_{n}\}_{n\geq
1}$ is pointwise
convergent. Let $D_{\infty}$ denote the limit function. Since $D_{n}$
is decreasing and convex for any $n\geq1$, $D_{\infty}$ is also decreasing
and convex. Moreover, by Fatou's lemma and (\ref{eqIntegralDn}),
\[
\int_{0}^{\infty}D_{\infty}(x)\,dx =\int
_{0}^{\infty}\liminf_{n\rightarrow\infty}D_{n}(x)\,dx
\leq\liminf_{n\rightarrow\infty}\int_{0}^{\infty}D_{n}(x)\,dx
=\overline{\mu}.
\]
In particular, since $D_{\infty}$ is decreasing and $\int_{0}^{\infty
}D_{\infty}(x)\,dx<\infty$,
we must have $\lim_{x\uparrow\infty}D_{\infty}(x)=0$.

Define a measure $\nu$ via $\nu((-\infty,x])=\mu(\R^+)-D_{\infty}(x)$.
It is clear that $\nu(\R^+)= \mu(\R^+)$, $F_{\nu}(0)=F_{\mu}(0)$,
$F_{\nu}(x)$ is continuous
and $\nu$ has a nonincreasing density $\rho$. Moreover, $\nu_n$
converges in distribution to $\nu$.

We wish to show that $\nu$ has mean $\overline{\mu}$. Note that $\nu_n
\preceq_{\mathrm{cx}} \breve{\mu}_n$ (Proposition~\ref{propProperties} applied
to $\mu_n$) and
$\mu_n \preceq_{\mathrm{cx}} \mu$ from which it follows that ${\nu}_n
\preceq
_{\mathrm{cx}} \breve{\mu}$. Hence, elements $\nu_n$ in the sequence have a
uniform bound (in the sense of convex order) and it follows that the
sequence is uniformly integrable and $\overline{\mu} = \lim_n
\overline
{\nu}_n = \overline{\nu}$.

Since $\nu_n$ converges to $\nu$ in distribution, it follows that
$P_{\nu_n}(x)$ converges pointwise to $P_\nu$.
Then since $\overline{\nu}_n \rightarrow\overline{\nu}$ it follows
that $C_{\nu_n}(x)$ converges pointwise to $C_\nu$.
Hence, $C_\nu(x) = \lim_n C_{\nu_n}(x) \geq\lim_n C_{\mu_n}(x) =
C_{\mu}(x)$.

Suppose $C_{\nu}(x)>C_{\mu}(x)$ on some interval $\mathcal{J}$: we show
that $\nu$ has constant density on $\mathcal{J}$. It is sufficient to
prove the result on every closed subinterval of $\mathcal{J}$, so we
assume $\mathcal{J} = [a,b]$. Then, by continuity of $C_\nu$ and
$C_\mu
$, there exists $\varepsilon$ such that $C_\nu(x) \geq C_\mu(x) +
\varepsilon
$ on $\mathcal{J}$. Let $\kappa= -C_{\mu}'(a+) \leq\mu(\R^+)$; then
$C'_\mu(x+) \geq- \kappa$ for all $x \in\mathcal{J}$. Fix $K \in
{\mathbb N}$ such that $K> 2(b-a)\kappa/\varepsilon$ and set $J^K= \{ a_j
\}_{0 \leq j \leq K}$ where $a_j = a + (b-a)j/K$. Since there is
pointwise convergence, there exists $N_{0}>0$ such that $C_{\nu
_{n}}(x)>C_{\mu}(x)+\varepsilon/2$
for all $x \in J^K$ and all $n\geq N_{0}$.
Then, for $n \geq N_0$, if $a_{j-1} \leq x \leq a_j$,
\begin{eqnarray*}
C_\mu(x) &\leq& C_\mu(a_{j-1})
\leq
C_\mu(a_j) + \kappa(a_j -
a_{j-1}) < C_\mu(a_j) + \varepsilon/2 <
C_{\nu_n}(a_j)
\\
&\leq& C_{\nu_n}(x).
\end{eqnarray*}
Finally, $C_{\mu_n}(x) \leq C_\mu(x)$ everywhere, and we conclude that
for sufficiently large $n$, $C_{\mu_n}(x) < C_{\nu_n}(x)$ on
$\mathcal
{J}$, and hence
$D_{n}(x)$ is a linear function on $\mathcal{J}$. It is easy to see
that $D_{\infty}(x)=\lim_{n\uparrow\infty}D_{n}(x)$
is also linear on $\mathcal{J}$, and thus $\nu$ has constant density on
$\mathcal{J}$ as required. Thus, the density of $\nu$ only
decreases when $C_{\nu}(x)=C_{\mu}(x)$. Hence, $\nu\in\sA^*_\mu$.
\end{pf}

\section{Uniqueness of a Nash equilibrium}
\label{secUnique}

Section~\ref{secExistence} shows that set $\mathcal{A}_{\mu}^{*}$
is nonempty. In this section, we prove that $\llvert\mathcal{A}_{\mu
}^{*}\rrvert\leq1$,
which thus completes the proof of Theorem \ref{thmSingleton}.

\begin{proposition} Suppose $\mu\in\sM$. Then $\llvert\sA^*_\mu
\rrvert\leq1$.
\label{proplessthanone}
\end{proposition}

\begin{pf}
Assume to the contrary that there exists two distinct elements $\nu$
and $\sigma$ in the set $\mathcal{A}_{\mu}^{*}$.
Recall that if $C_{\nu}(x)>C_{\mu}(x)$ then
$C_{\nu}$ is locally a quadratic function near $x$, similarly for
$C_{\sigma}$. Moreover, both $C_{\nu}'$ and $C_{\sigma}'$
are concave.

Observe that, for any $x\geq0$, we cannot have $C_{\nu}(y)>C_{\sigma}(y)$
for all $y\in(x,\infty)$: if so then $C_{\nu}(y)>C_{\sigma}(y)\geq
C_{\mu}(y)$
on $(x,\infty)$ and $C_{\nu}$ is quadratic on $(x,\infty)$, which
is impossible by Proposition \ref{propQuadBound}.

Let $x_0>0$ be such that $C_{\nu}(x_0)\neq C_{\sigma}(x_0)$. Without loss
of generality, suppose that $C_{\nu}(x_0)>C_{\sigma}(x_0)$. Define
$x_{1}=\inf\{ x>x_0\dvtx C_{\nu}(x)=C_{\sigma}(x) \} $.
By the observation above $x_{1}<\infty$. Also note that $C_{\nu
}(x)>C_{\sigma}(x)$
for all $x\in[x_0,x_{1})$.

Suppose that $C_{\nu}(x_{1})=C_{\sigma}(x_{1})>C_{\mu}(x_{1})$. Then,
near $x_{1}$,
\[
\cases{
\displaystyle C_{\nu}(x)=C_{\nu}(x_{1})+\beta_{\nu,1}(x-x_{1})+\gamma
_{\nu,1}(x-x_{1})^{2},
\cr
\displaystyle C_{\sigma}(x)=C_{\sigma}(x_{1})+\beta_{\sigma
,1}(x-x_{1})+ \gamma_{\sigma,1}(x-x_{1})^{2},}
\]
for some constants $\beta_{\nu,1}<0$, $\beta_{\sigma,1}<0$, $\gamma
_{\nu,1}>0$
and $\gamma_{\sigma,1}>0$. Since $C_{\nu}(x)>C_{\sigma}(x)$
to the left of $x_{1}$, it is clear that $\beta_{\nu,1}\leq\beta
_{\sigma,1}$.

Assume that $\beta_{\nu,1}=\beta_{\sigma,1}$. Then since $C_{\nu
}(x)>C_{\sigma}(x)$
on $[x_0,x_{1})$, we have $\gamma_{\nu,1}>\gamma_{\sigma,1}$. Let
$\hat{x}_{\nu,1}=\inf\{ x>x_{1}\dvtx C_{\nu}(x)=C_{\mu}(x)
\} $,
then $C_{\nu}(x)=C_{\nu}(x_{1})+\beta_{\nu,1}(x-x_{1})+\gamma_{\nu
,1}(x-x_{1})^{2}$
on $[x_{1},\hat{x}_{\nu,1}]$ and $\hat{x}_{\nu,1}<\infty$
by Proposition \ref{propQuadBound}. Further, by Proposition \ref
{propQuadBound},
$C_{\sigma}(x)\leq C_{\sigma}(x_{1})+\beta_{\sigma,1}(x-x_{1})+\gamma
_{\sigma,1}(x-x_{1})^{2}$
on $(x_{1},\infty)$. Thus, $C_{\sigma}(\hat{x}_{\nu,1})<C_{\nu
}(\hat
{x}_{\nu,1})=C_{\mu}(\hat{x}_{\nu,1})$,
which is a contradiction. Hence, $\beta_{\nu,1}<\beta_{\sigma,1}<0$.
Now let $\hat{x}_{\sigma,1}=\inf\{ x>x_{1}\dvtx C_{\sigma
}(x)=C_{\mu
}(x) \} <\infty$ (by Proposition~\ref{propQuadBound}).
If $\gamma_{\nu,1}\leq\gamma_{\sigma,1}$ then $C_{\nu}(\hat
{x}_{\sigma,1})<C_{\sigma}(\hat{x}_{\sigma,1})=C_{\mu}(\hat{x}_{\sigma,1})$,
which is a contradiction. So we conclude that $\beta_{\nu,1}<\beta
_{\sigma,1}<0$
and $\gamma_{\nu,1}>\gamma_{\sigma,1}>0$. Set $\vartheta=\beta
_{\sigma,1}-\beta_{\nu,1}>0$.

Now we introduce a useful lemma.

\begin{lemma} \label{lemxk}
Suppose $\nu$ and $\sigma$ are distinct elements of $\sA^*_\mu$.
Suppose $x_{k}$ is such that $C_{\nu}(x_{k})=C_{\sigma}(x_{k})>C_{\mu
}(x_{k})$.
Then $x_{k}>0$ and in a neighbourhood of $x_{k}$ we can write
\[
\cases{
\displaystyle C_{\nu}(x)=C_{\nu}(x_{k})+
\beta_{\nu,k}(x-x_{k})+\gamma_{\nu,k}(x-x_{k})^{2},
\cr
\displaystyle C_{\sigma}(x)=C_{\sigma}(x_{k})+\beta_{\sigma,k}(x-x_{k})+
\gamma_{\sigma,k}(x-x_{k})^{2}.}
\]
Suppose $\beta_{\nu,k} \neq\beta_{\sigma,k}$. Then
there is an interval to the left of $x_{k}$ on which $C_{\nu
}(x)-C_{\sigma}(x)$
is either strictly positive or strictly negative. Suppose that \mbox{$
C_{\nu}(x)-C_{\sigma}(x)>0$}
on some interval $(x_{k}-\varepsilon,x_{k})$: if not then interchange
the roles of $\nu$ and $\sigma$. Then $\beta_{\nu,k}<\beta_{\sigma,k}<0$
and $\gamma_{\nu,k}>\gamma_{\sigma,k}>0$.

Define $x_{k+1}=\sup\{ x<x_{k}\dvtx C_{\nu}(x)=C_{\sigma}(x)
\} $.
Then, $0< x_{k+1} < x_k$,\break $C_{\nu}(x_{k+1})=C_{\sigma
}(x_{k+1})>C_{\mu
}(x_{k+1})$,
and hence in a neighbourhood of $x_{k+1}$ we can write
%
\begin{equation}\label{eqRandnearXk+1}
\cases{
C_{\nu}(x)=C_{\nu}(x_{k+1})+\beta_{\nu,k+1}(x-x_{k+1})+\gamma_{\nu
,k+1}(x-x_{k+1})^{2},
\cr
C_{\sigma}(x)=C_{\sigma}(x_{k+1})+\beta_{\sigma,k+1}(x-x_{k+1})+
\gamma_{\sigma,k+1}(x-x_{k+1})^{2}.}
\end{equation}
Further, $\beta_{\sigma,k+1}<\beta_{\nu,k+1}<0$, $\gamma_{\sigma
,k+1}>\gamma_{\nu,k+1}>0$
and $ \beta_{\nu,k+1}-\beta_{\sigma,k+1}>\beta_{\sigma,k}-\beta
_{\nu,k}>0$.
\end{lemma}

\begin{pf}
Exactly as in the case $k=1$ from the proof of Proposition \ref
{proplessthanone},
we conclude that $\beta_{\nu,k}<\beta_{\sigma,k}<0$ and $\gamma
_{\nu,k}>\gamma_{\sigma,k}>0$.

Assume that $C_{\nu}(x_{k+1})=C_{\sigma}(x_{k+1})=C_{\mu}(x_{k+1})$.
If $x_{k+1}=0$, then\break $C'_{\nu}(x_{k+1}) = F_{\mu}(0)-\mu(\R^+) =
C'_{\sigma}(x_{k+1})$. Otherwise, if $x_{k+1}>0$ then\break
$C_{\nu}'(x_{k+1})=C_{\sigma}'(x_{k+1})$ by Proposition
\ref{propSmooth}. In either case, since $C_{\nu}(x)>\break C_{\sigma
}(x)\geq
C_{\mu}(x)$ on $(x_{k+1},x_{k})$,
$C_{\nu}'$ is linear on $[x_{k+1},x_{k}]$. Then because $\beta
_{\sigma,k} = C'_\sigma(x_k) > C'_{\nu}(x_k) = \beta_{\nu,k}$ and
$C_{\sigma}'$
is concave, it is clear that $C_{\sigma}'(x)\geq C_{\nu}'(x)$
on $(x_{k+1},x_{k})$. This contradicts the fact that $C_{\sigma}(\cdot)
= C_{\nu}(\cdot)$ at both $x_k$ and $x_{k+1}$.
Hence, $C_{\nu}(x_{k+1})=C_{\sigma}(x_{k+1})>C_{\mu}(x_{k+1})$.

It then follows that $x_{k+1}\in(0,x_{k})$ and both $C_{\nu}$ and
$C_{\sigma}$ are quadratic in a neighbourhood of $x_{k+1}$. In particular,
$C_{\nu}$ is quadratic on $(x_{k+1},x_{k})$ and $\gamma_{\nu,k}=\gamma
_{\nu,k+1}$.
Using a similar argument to that described in the case $k=1$, we get
that $\beta_{\sigma,k+1}<\beta_{\nu,k+1}<0$
and $\gamma_{\sigma,k+1}>\gamma_{\nu,k+1}>0$.

\begin{figure}[b]

\includegraphics{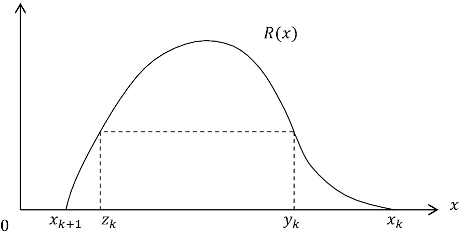}

\caption{Graph of $R(x)$. Since $R''(x)=\rho_{\nu}(x)-\rho_{\sigma}(x)$
is nondecreasing
on $(x_{k+1},x_{k})$ and $R''(y_{k})=0$, $R$ is concave on $[x_{k+1},y_{k}]$
and strictly convex on $[y_{k},x_{k}]$. Since $R'(x)=C_{\nu
}'(x)-C_{\sigma}'(x)$,
$R'(x_{k+1})=\beta_{\nu,k+1}-\beta_{\sigma,k+1}>0$ and
$R'(x_{k})=\beta
_{\nu,k}-\beta_{\sigma,k}<0$.
Then $R'(y_{k})<R'(x_{k})<0$, and there exists a unique $z_{k}\in
(x_{k+1},y_{k})$
such that $R(z_{k})=R(y_{k})$. Further, $R'(x_{k+1})\geq R'(z_{k})>0$.}
\label{figRx}
\end{figure}

Denote by $\rho_{\nu}$ the density function of the measure $\nu$. Then
$\rho_{\nu}$ is constant on $(x_{k+1},x_{k})$. In contrast, $\rho
_{\sigma}$
is nonincreasing, $\rho_{\sigma}(x_{k})=2\gamma_{\sigma,k}<2\gamma
_{\nu,k}=\rho_{\nu}(x_{k})$
and $\rho_{\sigma}(x_{k+1})=2\gamma_{\sigma,k+1}>2\gamma_{\nu,k+1}=\rho
_{\nu}(x_{k+1})$.
Set $y_{k}=\sup\{ y<x_{k}\dvtx\rho_{\sigma}(y)\geq\rho_{\nu
}(y)
\} $,
then $y_{k}\in(x_{k+1},x_{k})$. Further, $R(x)\triangleq C_{\nu
}(x)-C_{\sigma}(x)$
defined on $[x_{k+1},x_{k}]$ is zero at the endpoints, has nondecreasing
second derivative and is concave on $[x_{k+1},y_{k}]$ and strictly
convex on $[y_{k},x_{k}]$ (see Figure~\ref{figRx}).

Let $z_{k}\in(x_{k+1},y_{k})$ be the unique value such that
$R(z_{k})=R(y_{k})$.
Then $R'(x_{k+1})\geq R'(z_{k})>0$ and $R'(y_{k})<R'(x_{k})\leq0$.
Set $H(x)=R(y_{k}-x)-R(y_{k})$ on $[0,y_{k}-z_{k}]$. Then using Proposition
\ref{propIncrDSlope}, we obtain $R'(z_{k})\geq\llvert
R'(y_{k})\rrvert$.
Hence, $\beta_{\nu,k+1}-\beta_{\sigma,k+1}=R'(x_{k+1})\geq
R'(z_{k})\geq
\llvert R'(y_{k})\rrvert>\llvert R'(x_{k})\rrvert=\beta
_{\sigma,k}-\beta_{\nu,k}>0$.
\end{pf}

Return to the proof of Proposition \ref{proplessthanone}. Using Lemma
\ref{lemxk},
we construct a decreasing sequence of points $(x_{k})_{k\geq1}$ at
which $C_{\nu}-C_{\sigma}$ changes sign. Moreover, $\llvert C_{\nu
}'(x_{k})-C_{\sigma}'(x_{k})\rrvert=\llvert\beta_{\nu,k}-\beta
_{\sigma,k}\rrvert\geq\vartheta$.
Let $x_{\infty}=\lim_{k\uparrow\infty}x_{k}$ then $x_{\infty}\geq0$.
Observe that $\lim_{k\uparrow\infty}(\beta_{\nu,k}-\beta_{\sigma
,k})=C_{\nu}'(x_{\infty})-C_{\sigma}'(x_{\infty})$
exists. However, $\limsup_{k\uparrow\infty}(\beta_{\nu,k}-\beta
_{\sigma,k})\geq\vartheta>0$
and $\liminf_{k\uparrow\infty}(\beta_{\nu,k}-\beta_{\sigma,k})\leq
-\vartheta<0$,
which is a contradiction. Hence, there cannot be distinct elements
$\nu$ and $\sigma$ in set $\mathcal{A}_{\mu}^{*}$.

The above is predicated on the assumption that $C_{\nu
}(x_{1})=C_{\sigma
}(x_{1})>C_{\mu}(x_{1})$.
Now suppose $C_{\nu}(x_{1})=C_{\sigma}(x_{1})=C_{\mu}(x_{1})$. Recall
that $C{}_{\nu}>C_{\sigma}$ on an interval to the left of $x_{1}$.
Let $x_{2}=\sup\{ x<x_{1}\dvtx C_{\nu}(x)=C_{\sigma}(x) \} $.
Then $x_{2}\in[0,x_0)$ and $C{}_{\nu}(x)>C_{\sigma}(x)$ on $(x_{2},x_{1})$.

Assume that $C_{\nu}(x_{2})=C_{\sigma}(x_{2})=C_{\mu}(x_{2})$. Then,
by Proposition \ref{propSmooth},\break $C_{\nu}'(x_{1})=C_{\sigma}'(x_{1})$
and $C_{\nu}'(x_{2})=C_{\sigma}'(x_{2})$.
Since $C_{\nu}'$ is linear and $C_{\sigma}'$ is concave
on $(x_{2},x_{1})$, it follows that $C_{\nu}'(x)\leq C_{\sigma}'(x)$
for all
$x\in[x_{2},x_{1}]$, which is a contradiction. Hence, $C_{\nu
}(x_{2})=C_{\sigma}(x_{2})>C_{\mu}(x_{2})$.

Now starting the construction at $x_{2}$, rather than $x_{1}$, we
are in the same case as discussed previously. In particular, there
cannot be distinct elements $\nu$ and $\sigma$ in set~$\mathcal
{A}_{\mu}^{*}$.\vspace*{-9pt}
\end{pf}

\begin{appendix}\label{secProveAtomfree}
\section*{Appendix: Necessity of conditions for a symmetric Nash~equilibrium}\vspace*{-9pt}

\begin{pf*}{Proof of Theorem \ref{thmAtomfree}}
First, we argue that for $\nu\in\sP$ to be a symmetric Nash
equilibrium we must have that $\overline{\nu}=\overline{\mu}$.
Suppose not. Then $\overline{\nu} < \overline{\mu}$.

Let $\alpha$ be the unique solution of $T(\cdot)=0$ where $T(x) =
P_{\mu
}(x) - x + \overline{\nu}$. There are two cases; either $P_\nu
(\alpha)
= P_{\mu}(\alpha) = \alpha- \overline{\nu}$ or $P_\nu(\alpha) >
P_{\mu
}(\alpha)$. We consider the second case first; the first case is
degenerate and will be treated subsequently.

So suppose $P_\nu(\alpha) > P_{\mu}(\alpha)$. Then $P_\nu(x) \geq
x -
\overline{\nu} > P_{\mu}(x)$ on $(\alpha, \infty)$ and
$\varepsilon:= \inf_{x \geq\alpha} [P_{\nu}(x) - P_{\mu}(x)] > 0$. The
aim is to construct an interval $(\gamma,\beta)$ and to modify $\nu$ by
moving the mass of $\nu$ in this interval to the right-hand endpoint,
whilst preserving the admissibility property. Then the modified measure
is preferred to $\nu$ when playing against an agent with target law
$\nu$.

Consider $Q(\cdot)$ defined on $x \geq\alpha$ by $Q(x) = P_{\nu}(x) -
\varepsilon$. Then $Q$ is nonnegative. Fix $\beta> \alpha$ and define
$\gamma$ via
\[
\gamma= \mathop{\arg\sup}_{c < \beta} \biggl\{ \frac{Q(\beta) - P_{\nu
}(c)}{\beta- c} \biggr\},\qquad
\Gamma= \frac{Q(\beta) -
P_{\nu}(\gamma)}{\beta- \gamma}.
\]
(Note $\gamma$ may not be uniquely defined, but $\Gamma$ is.) Then
$F_\nu(\gamma-) \leq\Gamma\leq F_\nu(\gamma) < F_{\nu}(\beta-)$.

Let $P_\sigma$ be defined by
\[
P_{\sigma}(x) = \cases{ P_{\nu}(x), &\quad$0 \leq x < \gamma$;
\cr
P_{\nu}(\gamma) + (x - \gamma) \Gamma, &\quad$\gamma\leq x < \beta$;
\cr
Q(x), &\quad$x \geq\beta$.}
\]
If $\sigma$ is the measure associated with the put price function
$P_\sigma$ then $\sigma\geq\nu$ in the sense of first-order
stochastic dominance, and $P_\sigma\geq P_\mu$ so that $\sigma$ is
weakly admissible.

For any admissible strategy $\pi$,
\begin{eqnarray*}
V_{\sigma,\pi}^{1}-V_{\nu,\pi}^{1} & = & \int
_{0}^{\infty}\pi(dx) \bigl\{ \bigl[
\bigl(1-F_{\sigma}(x) \bigr)+\theta\bigl(F_{\sigma}(x)-F_{\sigma}(x-)
\bigr) \bigr]
\\
&&\hspace*{51pt}{} - \bigl[ \bigl(1-F_{\nu}(x) \bigr)+\theta
\bigl(F_{\nu}(x)-F_{\nu}(x-) \bigr) \bigr] \bigr\}
\\
& = & \pi\bigl(\{\gamma\} \bigr) (1-\theta) \bigl[F_{\nu}(\gamma)-
\Gamma\bigr]
\\
&&{} +\int_{(\gamma,\beta)}\pi(dx) \bigl\{(1-\theta)
\bigl[F_{\nu
}(x)-\Gamma\bigr]+\theta\bigl[F_{\nu}(x-)-\Gamma
\bigr] \bigr\}
\\
&&{} +\pi\bigl(\{\beta\} \bigr)\theta\bigl(F_{\nu}(
\beta-)- \Gamma\bigr)
\\
& \geq& (1-\theta) \biggl\{ \pi\bigl(\{\gamma\} \bigr) \bigl[F_{\nu}(
\gamma)-\Gamma\bigr]+\int_{(\gamma,\beta)}\pi(dx) \bigl[F_{\nu}(x)-
\Gamma\bigr] \biggr\}.
\end{eqnarray*}
Then since $\nu$ assigns positive mass to $[\gamma,\beta)$ it follows
that $V^1_{\sigma, \nu} - V^1_{\nu,\nu} > 0$ and $(\nu,\nu)$
cannot be
a symmetric Nash equilibrium.

Now consider the degenerate case $P_\nu(\alpha)= \alpha- \overline
{\nu
}$. Then $\nu$ has support on $[0,\alpha]$ and an atom at $\alpha$ and
$\mu$ assigns positive mass to $(\alpha,\infty)$.

Define $\sigma$ via the put price function $P_\sigma$ where
$P_{\sigma}(x) = P_{\nu}(x)$ for $x \leq\alpha$ and $P_{\sigma}(x)=
P_{\mu}(x)$ for $x > \alpha$. Note that $P_{\nu}$ and $P_{\mu}$ are
convex and $P'_{\sigma}(\alpha-) = P'_{\nu}(\alpha-) \leq P'_{\mu
}(\alpha-) \leq P'_{\mu}(\alpha+) = P'_{\sigma}(\alpha+)$, and hence
$P_{\sigma}$ is convex.
Then also $\sigma$ is admissible and
\begin{eqnarray*}
V_{\sigma,\pi}^{1}-V_{\nu,\pi}^{1} & =&\pi\bigl(\{
\alpha\} \bigr) (1-\theta) \bigl[1-F_{\mu
}(\alpha) \bigr]+\int
_{(\alpha,\infty)}\pi(dx)\bigl\{ \bigl[1-F_{\mu}(x) \bigr]\bigr\}
\\
&&{} +\theta\bigl[F_{\mu}(x)-F_{\mu}(x-) \bigr]
\end{eqnarray*}
and $V^1_{\sigma, \nu} - V^1_{\nu,\nu} > \nu(\{\alpha\})(1-\theta
) \mu
((a,\infty))> 0$. We conclude in this case also that $(\nu,\nu)$ cannot
be a symmetric Nash equilibrium.

Second, we show that for $\nu$ to be a symmetric Nash equilibrium we
must have that $\nu$ has no atoms at points above zero. Assume that
${\nu}$ is strongly admissible and that $\nu$ places an atom of size
$p>0$ at $z>0$. We aim to show that $\nu$ cannot correspond to a
symmetric Nash equilibrium by considering the impact of splitting the
mass at $z$ into a mass of size $q$ at $z - \varepsilon_1$ and a mass of
size $p-q$ at $z+\varepsilon_2$ where $q \ll p $ and $\varepsilon_1 \gg
\varepsilon_2$, in such a way that the mean is preserved.

Let the measure $\sigma$ be given by
\[
F_{\sigma}(x)= %
\cases{ F_{\nu}(x), &\quad if $x
\in[0,z-\varepsilon_{1})\cup[z+\varepsilon_{2},\infty)$,
\cr
F_{\nu}(x)+q, &\quad if $x\in[z-\varepsilon_{1},z)$,
\cr
F_{\nu}(x)-(p-q), &\quad if $x\in[z,z+\varepsilon_{2})$,}
\]
where $\varepsilon_{2}\in(0,\frac{(1-\theta)zp}{1+\theta
p} )$,
$\varepsilon_{1}\in(\frac{(1+\theta p)\varepsilon_{2}}{(1-\theta
)p},z )$
and $q=\frac{\varepsilon_{2}p}{\varepsilon_{1}+\varepsilon_{2}}\in(0,p)$.
Observe that $(z-\varepsilon_{1})q+(z+\varepsilon_{2})(p-q)=zp$, and hence
$F_{\nu}$ and $F_{\sigma}$ have the same mean. Then
$C_{\sigma}(x)=C_{\nu}(x)$ if $x\in[0,z-\varepsilon_{1})\cup
[z+\varepsilon
_{2},\infty)$,
$C_{\sigma}(x)=C_{\nu}(x)+[x-(z-\varepsilon_{1})]q$ if $x\in
[z-\varepsilon_{1},z)$,
and $C_{\sigma}(x)=C_{\nu}(x)+(z+\varepsilon_{2}-x)(p-q)$ if $x\in
[z,z+\varepsilon_{2})$.
This implies that $C_{\sigma}(x)\geq C_{\nu}(x)\geq C_{\mu}(x)$.
Thus, $\sigma$ is strongly admissible.

Suppose that Player 2 chooses law ${\nu}$. Then
\begin{eqnarray*}
V_{\sigma,\nu}^{1}-V_{\nu,\nu}^{1} &
=&F_{\nu} \bigl((z-\varepsilon_{1})- \bigr)q+F_{\nu
}
\bigl((z+\varepsilon_{2})- \bigr) (p-q)-F_{\nu}(z-)p-\theta
p^{2}
\\
&&{} +\theta\nu\bigl(\{z-\varepsilon_{1}\} \bigr)q+\theta\nu
\bigl(\{ z+\varepsilon_{2}\} \bigr) (p-q)
\\
& \geq& F_{\nu} \bigl((z-\varepsilon_{1})- \bigr)q+F_{\nu}
\bigl((z+\varepsilon_{2})- \bigr) (p-q)-F_{\nu}(z-)p-\theta
p^{2}
\\
& =&p \biggl\{ \bigl[F_{\nu} \bigl((z+\varepsilon_{2})-
\bigr)-F_{\nu}(z-) \bigr]\frac
{\varepsilon_{1}}{\varepsilon_{1}+\varepsilon_{2}}
\\
&&\hspace*{11pt}{} - \bigl[F_{\nu}(z-)-F_{\nu} \bigl((z-
\varepsilon_{1})- \bigr) \bigr]\frac
{\varepsilon_{2}}{\varepsilon_{1}+\varepsilon_{2}}-\theta p \biggr\}.
\end{eqnarray*}
Since $F_\nu((z+\varepsilon_{2})-)-F_\nu(z-)\geq p$ and $0 \leq F_\nu
(z-)-F_\nu((z-\varepsilon_{1})-)\leq1$,
\[
V^1_{\sigma,\nu}-V^1_{\nu,\nu}\geq p \biggl[
\frac{\varepsilon
_{1}}{\varepsilon
_{1}+\varepsilon_{2}}p-\frac{\varepsilon_{2}}{\varepsilon
_{1}+\varepsilon
_{2}}-\theta p \biggr]=p \frac{\varepsilon_{1}(1-\theta)p-(1+\theta
p)\varepsilon_{2}}{\varepsilon_{1}+\varepsilon_{2}}>0,
\]
which contradicts the assumption that $(\nu,\nu)$ is a Nash
equilibrium. Thus, $F_{\nu}(x)$
is continuous on $(0,\infty)$.

Third, we consider the possibility of an atom at zero. Suppose $(\nu,\nu
)$ is a symmetric Nash equilibrium and set $p=F_{\nu}(0)$ and $p_{\mu
}=F_{\mu}(0)$. Since $\nu$ must be strongly (and not merely weakly)
admissible by the first part of the theorem, $\mu\preceq_{\mathrm{cx}}\nu$ and
we must have $p\geq p_{\mu}$. Suppose that $p>p_{\mu}$; we aim to
derive a contradiction. Fix any $q$ such that
$0<q<\min\{ p\sqrt{1-\theta},1-p \} $. Since by the arguments
above we must have that $F_{\nu}$
is continuous on $(0,\infty)$, there exists $\varepsilon>0$ such that
$\nu((0,\varepsilon))=q$, and then $F_{\nu}(\varepsilon)=p+q$. For any
$\phi
\in(0,1)$,
let measure $\sigma_{\phi}$ be given by
\[
F_{\sigma_{\phi}}(x)= %
\cases{ (1-\phi)F_{\nu}(x), &\quad if
$x\in[0,\delta)$,
\cr
\phi(p+q)+(1-\phi)F_{\nu}(x), &\quad if $x\in[
\delta,\varepsilon)$,
\cr
F_{\nu}(x), &\quad if $x\in[\varepsilon,\infty)$,}
\]
where $\delta=\int_{0}^{\varepsilon}y\nu(dy)/(p+q)$. Then $\sigma
_{\phi}$
is a probability measure with the same mean as $\nu$.
It follows that
\begin{eqnarray*}
V^1_{\sigma_{\phi},\nu}-V^1_{\nu,\nu} & =&\phi\biggl\{
(p+q)F_{\nu
}(\delta)-\theta p^{2}-\int_{0}^{\varepsilon}F_{\nu}(y)
\nu(dy) \biggr\}
\\
& \geq&\phi\bigl\{ (p+q)p-\theta p^{2}-(p+q)q \bigr\} =\phi\bigl\{
(1-\theta)p^{2}-q^{2} \bigr\} >0.
\end{eqnarray*}
Hence, if $\sigma_{\phi}$ is strongly admissible then Player 1 would prefer
strategy $\sigma_{\phi}$ to $\nu$.

Making $q$ and $\varepsilon$ smaller if necessary, and using the fact
that $C_{\nu}'(0+)=p-1>p_{\mu}-1=C_{\mu}'(0+)$, we can insist
that $C_{\nu}(x)-C_{\mu}(x)>(p-p_{\mu})x/2$ for $x\in(0,\varepsilon)$.
Observe that $C_{\nu}(x)-C_{\sigma_{\phi}}(x)=0$ for $x\geq\varepsilon$.
Moreover, for $x\in[0,\varepsilon)$, since $F_{\nu}(x)-F_{\sigma_{\phi
}}(x)\leq\phi F_{\nu}(\delta)$,
$C_{\nu}(x)-C_{\sigma_{\phi}}(x)\leq\phi F_{\nu}(\delta)x$. Then,
if $\phi\leq\frac{p-p_{\mu}}{2F_{\nu}(\delta)}$, we have
\begin{eqnarray*}
C_{\sigma_{\phi}}(x)-C_{\mu}(x) & = &\bigl(C_{\nu}(x)-C_{\mu
}(x)
\bigr)- \bigl(C_{\nu}(x)-C_{\sigma_{\phi}}(x) \bigr)
\\
& >&\tfrac{1}{2}(p-p_{\mu})x-\phi F_{\nu}(\delta)x\geq0
\end{eqnarray*}
for all $x\in(0,\varepsilon)$, and thus $\mu\preceq_{\mathrm{cx}}\sigma_{\phi}$.
Hence, $\sigma_\phi$ is admissible for small enough $\phi$ and $(\nu,\nu
)$ cannot be a symmetric Nash equilibrium.
It follows that $p=p_\mu$ and $F_{\nu}(0)=F_{\mu}(0)$.
\end{pf*}

\begin{pf*}{Proof of forward implication of Theorem \ref{thmcharacterisation}}
We have shown that if $\nu$ is a symmetric Nash equilibrium then $\nu$
must be strongly admissible with respect to $\mu$. It remains to show,
first that $\nu$ must have a decreasing density, and second that the
density can only decrease at points where the convex order constraint
is binding.

Let $(\pi,\pi)$ be a candidate symmetric Nash equilibrium, and suppose
$\pi$ has the properties given in Theorem~\ref{thmAtomfree}.
Suppose that $\pi$ is such that $F_{\pi}$ is not concave on
$(0,\infty
)$. Then there exist $a,b$ with $0<a<b< \infty$ such that
$F_\pi(b) - F_\pi(a)>0$ and for $x \in(a,b)$,
\[
F_\pi(x) < F_{\pi}(a) + \frac{x-a}{b-a} \bigl[
F_{\pi}(b)-F_{\pi}(a) \bigr].
\]
Let $\sigma$ be such that $F_{\sigma}(x) = F_\pi(x)$ for $x$ outside
$[a,b)$ and for $x \in[a,b)$,
$F_\sigma(x) = F_\pi(a) + \phi[ F_{\pi}(b) - F_{\pi}(a)]$, where
$\phi
$ is chosen so that the means of $\sigma$ and $\pi$ agree.
Then $\int_a^b (F_\pi(x)-F_\pi(a)) \,dx = \int_a^b (F_\sigma
(x)-F_\pi(a))
\,dx = \phi[ F_{\pi}(b) - F_{\pi}(a)] (b-a)$ and it follows that
$\phi<1/2$.
Then
\begin{eqnarray*}
V_{\sigma,\pi}^{1}-V_{\pi,\pi}^{1} & = & \int
_{[a,b)}F_{\pi
}(x) \bigl[\sigma(dx)-\pi(dx) \bigr]
\\
& = & \int_{[a,b)} \bigl(F_{\pi}(x)-F_{\pi}(a)
\bigr) \bigl[\sigma(dx)-\pi(dx) \bigr]
\\
& = & (1-\phi) \bigl(F_{\pi}(b)-F_{\pi}(a)
\bigr)^{2}-\frac{(F_{\pi
}(x)-F_{\pi}(a))^{2}}{2}\bigg|_{x=a}^{b}
\\
& = & \biggl(\frac{1}{2}-\phi\biggr) \bigl(F_{\pi}(b)-F_{\pi}(a)
\bigr)^{2}>0,
\end{eqnarray*}
and hence $(\pi,\pi)$ cannot be a symmetric Nash equilibrium.

Now suppose that $\nu$ is strongly admissible, has no atoms on
$(0,\infty)$ and $F_\nu$ is concave on $(0,\infty)$. Then the density
$f_\nu\triangleq F_\nu'$ is decreasing. Without loss of generality we
take $f_\nu$ to be right-continuous. Suppose that $z$ is such that
$f_\nu(x) > f_\nu(z)$ for all $x<z$ and that $C_{\nu}(z)>C_{\mu}(z)$.
Then there exists $\varepsilon\in(0,z)$ such that $C_{\nu}(z-\varepsilon
) +
2 \varepsilon C'_{\nu}(z - \varepsilon) > C_{\mu}(z+\varepsilon)$ and
$C_{\nu}(z+\varepsilon) - 2 \varepsilon C'_{\nu}(z + \varepsilon) >
C_{\mu
}(z-\varepsilon)$, and it follows that if $\sigma$ is any measure such
that $C_{\sigma} = C_\nu$ outside $(z - \varepsilon,z+\varepsilon)$ then
$C_{\sigma}(x) > C_\mu(x)$ on the interval $[z-\varepsilon,z+\varepsilon]$
and $\sigma$ is strongly admissible.

Given $z$ and $\varepsilon$ as in the previous paragraph, let $A_+ =
\int_z^{z+\varepsilon} (F_{\nu}(z+\varepsilon) - F_\nu(x)) \,dx$ and let $A_{-}(x)
= \int_x^{F_\nu(z)}
(z - F_{\nu}^{-1}(u)) \,du$. Let $v$ solve $A_{-}(v) = A_+$ and set $w =
F_{\nu}^{-1}(v)$. Note that $A_{-}(F_\nu(z-\varepsilon)) \geq A_+$ by the
concavity of $F_\nu$. Then $w \geq z -\varepsilon$ and $v \geq F_{\nu
}(0)$. Note further that $A_{-}(F_{\nu}(z) - \{ F_{\nu}(z+\varepsilon) -
F_\nu(z) \}) < f_\nu(z)\{ F_{\nu}(z+\varepsilon) - F_\nu(z) \}/2 \leq A_+$
and, therefore, $v < 2 F_\nu(z) - F_\nu(z + \varepsilon)$, see
Figure~\ref
{figConcavity}.

\begin{figure}

\includegraphics{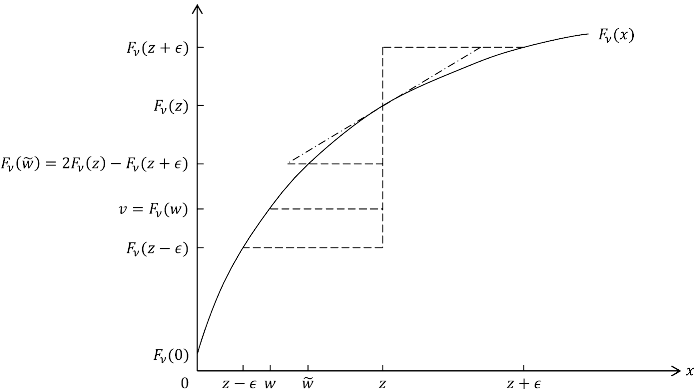}

\caption{Plot of the concave function $F_{\nu}(\cdot)$. By concavity of
$F_\nu$ we have $w > z-\varepsilon$. If $F_{\nu}(\tilde{w}) = F_{\nu
}(z) -
\{ F_{\nu}(z+\varepsilon) - F_\nu(z) \}$, then the area $A_-(F_{\nu
}(\tilde
{w}))$ is strictly less than the area of the triangle which lies above
the horizontal line at $F_\nu(\tilde{w})$, to the left of the vertical
line at $z$ and below the tangent to $F_\nu$ at $z$ with slope $f_\nu
(z)$, as represented by the sloping line. This area is equal to the
area of the triangle delimited by the vertical line at $z$, the
horizontal line at $F_{\nu}(z+\varepsilon)$ and the same tangent to
$F_\nu
$. In turn, this area is less than or equal to $A_+$. It follows that
$w < \tilde{w}$ and $v=F_\nu(w) < F_\nu(\tilde{w}) = 2F_{\nu}(z) -
F_{\nu}(z+\varepsilon)$.}
\label{figConcavity}
\end{figure}

Then by construction, $\int_w^{z+\varepsilon} x \nu(dx) = z \int
_w^{z+\varepsilon} \nu(dx)$, and if we define $\sigma$ by $F_\sigma=
F_\nu
$ outside $(w, z+\varepsilon)$ and $F_\sigma(x) = F_\nu(w)$ for $w < x <
z$ and $F_\sigma(x) = F_\nu(z+\varepsilon)$ for $z \leq x <
z+\varepsilon
$ we
have that $\sigma$ has the same mean as $\nu$. [In effect, $\sigma$
replaces the mass of $\nu$ on $(w, z+\varepsilon)$ with a point mass at
$z$.] Then, by the remarks of the previous paragraph, $\sigma$ is
strongly admissible with respect to $\mu$.

Finally,
\begin{eqnarray*}
V^1_{\sigma,\nu} - V^1_{\nu,\nu} & = & \int
_{(w,z+\varepsilon)} F_{\nu
}(x) \bigl[ \sigma(dx) - \nu(dx) \bigr]
\\
& = & \bigl(F_\nu(z+\varepsilon) - v \bigr) F_\nu(z) -
\frac{F_\nu(z+\varepsilon
)^2 -
v^2}{2}
\\
& = & \frac{(F_\nu(z+\varepsilon) - v)}{2} \bigl[ 2F_\nu(z) - v - F_\nu
(z+\varepsilon) \bigr] > 0
\end{eqnarray*}
and hence $(\nu,\nu)$ cannot be a symmetric Nash equilibrium.
\end{pf*}
\end{appendix}



%

\printaddresses
\end{document}